\documentclass[english]{article}
 
\makeatletter

\usepackage{hyperref} 
\newcommand{\STAB}[1]{\begin{tabular}{@{}c@{}}#1\end{tabular}}
\usepackage{amsmath,amssymb,amsthm,dsfont,mathrsfs}
\usepackage{makecell}
\usepackage{array,multirow,graphicx}
\usepackage{setspace}
\newcommand{\E}{\mathbb{E}}

\DeclareMathOperator*{\argmin}{argmin} 
\usepackage{algorithm}
\usepackage[noend]{algorithmic}
\usepackage{enumitem}

\theoremstyle{plain}

\theoremstyle{plain}

\theoremstyle{plain}
\newtheorem{lemma}{\protect\lemmaname}
\theoremstyle{plain}
\newtheorem{theorem}{\protect\theoremname}
\theoremstyle{plain}
  
\theoremstyle{definition}

\theoremstyle{definition}

\theoremstyle{definition}

\providecommand{\claimname}{Claim}
\providecommand{\lemmaname}{Lemma}
\providecommand{\propositionname}{Proposition}
\providecommand{\theoremname}{Theorem}
\providecommand{\corollaryname}{Corollary} 
\providecommand{\definitionname}{Definition}
\providecommand{\assumptionname}{Assumption}
\providecommand{\remarkname}{Remark}

\numberwithin{theorem}{section}
\numberwithin{lemma}{section}
\numberwithin{equation}{section}

\makeatother

\usepackage{geometry} 
\makeatletter
\newcommand{\manuallabel}[2]{\def\@currentlabel{#2}\label{#1}}
\makeatother

\begin{document} 

\title{Fast Splitting Algorithms for \\ Sparsity-Constrained and Noisy Group Testing}
\author{Eric Price, Jonathan Scarlett, and Nelvin Tan}
\date{}

\maketitle

\begin{abstract}
In group testing, the goal is to identify a subset of defective items within a larger set of items based on tests whose outcomes indicate whether at least one defective item is present. This problem is relevant in areas such as medical testing, DNA sequencing, communication protocols, and many more. In this paper, we study (i) a sparsity-constrained version of the problem, in which the testing procedure is subjected to one of the following two constraints: items are finitely divisible and thus may participate in at most $\gamma$ tests; or tests are size-constrained to pool no more than $\rho$ items per test; and (ii) a noisy version of the problem, where each test outcome is independently flipped with some constant probability. Under each of these settings, considering the for-each recovery guarantee with asymptotically vanishing error probability, we introduce a fast splitting algorithm and establish its near-optimality not only in terms of the number of tests, but also in terms of the decoding time.  While the most basic formulations of our algorithms require $\Omega(n)$ storage for each algorithm, we also provide low-storage variants based on hashing, with similar recovery guarantees.
\end{abstract}

\long\def\symbolfootnote[#1]#2{\begingroup\def\thefootnote{\fnsymbol{footnote}}\footnote[#1]{#2}\endgroup}

\symbolfootnote[0]{ E.~Price is with the Department of Computer Science, University of Texas at Austin (e-mail: \url{ecprice@cs.utexas.edu}). J.~Scarlett is with the  Department of Computer Science, the Department of Mathematics, and the Institute of Data Science, National University of Singapore  (e-mail: \url{scarlett@comp.nus.edu.sg}).  N.~Tan is  is with the  Department of Computer Science, National University of Singapore  (e-mail: \url{nelvintan@u.nus.edu}).

E.~Price was supported in part by NSF Award CCF-1751040 (CAREER).  J.~Scarlett was supported by an NUS Early Career Research Award. }

\vspace*{-2ex}
\section{Introduction}

In the group testing problem, the goal is to identify a small subset $\mathcal{S}$ of defective items of size $k$ within a larger set of items of size $n$, based on a number $T$ of tests. This problem is relevant in areas such as medical testing, DNA sequencing, and communication protocols \cite[Sec.~1.7]{Ald19}, and more recently, utility in testing for COVID-19 \cite{Hogan2020,Yelin2020}.

In this paper, we present algorithms for sparsity-constrained (bounded tests-per-item or items-per-test) and noisy variants of group testing with a near-optimal sublinear decoding time, building on techniques recently proposed for the unconstrained noiseless group testing problem \cite{cher20,Eri20}. 
These extensions come with new challenges presented by the infeasibility of the designs in \cite{cher20,Eri20} in the sparsity-constrained setting, and the need to handle both false positive and false negative tests in the noisy setting.

\subsection{Problem Setup}

Let $n$ denote the number of items, which we label as $\{1,\dots,n\}$. Let $\mathcal{S}\subset\{1,\dots,n\}$ denote the fixed set of defective items, and let $k=|\mathcal{S}|$ be the number of defective items.  To avoid cumbersome notation, we present our algorithms in a form that uses $k$ directly; however, the analysis goes through unchanged when an upper bound $\bar{k} \ge k$ is used instead, and $\bar{k}$ replaces $k$ in the number of tests and decoding time.

We are interested in asymptotic scaling regimes in which $n$ is large and $k$ is comparatively small, and thus assume that $k=o(n)$ throughout. We let $T=T(n)$ be the \textit{number of tests} performed. In the noiseless setting, the $i$-th test takes the form 
\begin{align}
    Y^{(i)}=\bigvee_{j\in\mathcal{S}}X_j^{(i)}, \label{eq:test_outcome_formula}
\end{align}
where the test vector $X^{(i)}=\big(X_1^{(i)},\dots,X_n^{(i)}\big)\in\{0,1\}^n$ indicates which items are are included in the test, and $Y^{(i)}\in\{0,1\}$ is the resulting observation, indicating whether at least one defective item was included in the test. The goal of group testing is to design a sequence of tests $X^{(1)},\dots,X^{(T)}$, with $T$ ideally as small as possible, such that the outcomes can be used to reliably recover the defective set $\mathcal{S}$ with probability close to one, while ideally also having a low-complexity decoding procedure. We focus on the non-adaptive setting, in which all tests $X^{(1)},\dots,X^{(T)}$ must be designed prior to observing any outcomes.

We consider the \textit{for-each} recovery guarantee; specifically, we seek to develop a randomized algorithm that, for any fixed defective set $\mathcal{S}$ of cardinality $k$, produces an estimate $\widehat{\mathcal{S}}$ such that the error probability $P_e:=\mathbb{P}\big[\widehat{\mathcal{S}}\neq\mathcal{S}\big]$ is asymptotically vanishing as $n \to \infty$. For all of our algorithms, only the tests $\big\{X^{(i)}\big\}_{i=1}^T$ will be randomized, and the decoding procedure will be deterministic given the test outcomes.

\textbf{Notation.} Throughout the paper, the function $\log(\cdot)$ has base $e$, and we make use of Bachmann-Landau asymptotic
notation (i.e., $O$, $o$, $\Omega$, $\omega$, $\Theta$), as well as the notation $\widetilde{O}(\cdot)$, which omits poly-logarithmic factors in its argument.

\subsubsection{Sparsity-Constrained Setting}

In the sparsity-constrained group testing problem \cite{Ven19}, the testing procedure is subjected to one of two constraints:
\begin{itemize}
    \item Items are \textit{finitely divisible} and thus may participate in at most $\gamma=o(\log n)$ tests;
    \item Tests are \textit{size-constrained} and thus contain no more than $\rho=o(n/k)$ items per test.
\end{itemize}
For instance, in the classical application of testing blood samples for a given disease \cite{Dor43}, the $\gamma$-divisible items constraint may arise when there are limitations on the volume of blood provided by each individual, and the $\rho$-sized test constraint may arise when there are limitations on the number of samples that the machine can accept, or on the number that can be mixed together while avoiding undesirable dilution effects.

It is well known that if each test comprises of $\Theta(n/k)$ items, then $\Theta(\min\{n,k\log n\})$ tests suffice for group testing algorithms with asymptotically vanishing error probability \cite{Cha14,Ald14a,Sca15b,Joh16}.  Moreover, this scaling is known to be optimal \cite{Bay20}.  Hence, the parameter regime of primary interest in the size-constrained setting is $\rho=o(n/k)$. By a similar argument, the parameter regime of primary interest in the finitely divisible setting is $\gamma=o(\log n)$.  

\subsubsection{Noisy Setting} \label{sec:noise_intro}

Generalizing \eqref{eq:test_outcome_formula}, we consider the following widely-adopted symmetric noise model:
\begin{align}
    Y^{(i)}=\bigg(\bigvee_{j\in\mathcal{S}}X_j^{(i)}\bigg)\oplus Z, \label{eq:noise}
\end{align}
where $Z\sim\text{Bernoulli}(p)$ for some $p\in(0,1/2)$, and $\oplus$ denotes modulo-2 addition. While the symmetry assumption may appear to be restrictive, our results and analysis will hold with essentially no change under any non-symmetric random noise model where $0 \to 1$ flips and $1 \to 0$ flips both have probability at most $p$.

Throughout the paper, we will focus {\em separately} on the sparsity-constrained aspects and noisy aspects.  While their joint treatment is also of interest, it was shown in \cite{Ven19} that for finitely divisible items, if the tests are subject to random noise of the form in \eqref{eq:noise}, then the error probability is bounded away from zero regardless of the total number of tests in the finitely-divisible setting with $\gamma = o(\log k)$.  Thus, at least in most scaling regimes of interest, handling noise and finite-divisibility constraints simultaneously would require changing the noise model and/or the recovery criteria, and we make no attempt to do so. On the other hand, for noisy size-constrained tests, schemes that attain asymptotically vanishing error probability do indeed exist \cite{Ven19}.  We still focus on the size-constrained and noisy aspects separately for clarity of exposition, but the two can be combined using our techniques in a straightforward manner, as we briefly discuss in Appendix \ref{sec:non_binary}.

\subsubsection{Mathematical and Computational Assumptions}
Throughout the paper, we assume a word-RAM model of computation; for instance, with $n$ items and $T$ tests, it takes $O(1)$ time to read a single integer in $\{1,\dots,n\}$ from memory, perform arithmetic operations on such integers, fetch a single test outcome indexed by $\{1,\dots,T\}$ and so on.

For simplicity of notation, we assume throughout the analysis that $k$, $n$, and $\rho$ are powers of two. Our algorithm only requires an upper bound on the number of defectives, and hence, any other value of $k$ can simply be rounded up to a power of two. In addition, the total number of items $n$ can be increased to a power of two by adding ``dummy'' non-defective items, and $\rho$ can be rounded down without impacting our final scaling laws (we do not seek to characterize the precise constants).

\subsection{Related Work}

While extensive works have studied the number of tests for various group testing strategies (see \cite{Ald19} for a survey), relatively fewer have sought efficient ${\rm poly}(k \log n)$ decoding time.  For the standard noiseless group testing problem, the most relevant existing results come from two recent concurrent works \cite{cher20,Eri20}, which showed that there exists a non-adaptive group testing algorithm that succeeds with $O(k\log n)$ tests and has $O(k\log n)$ decoding time. We build on these splitting techniques in this paper; the existing approach is outlined in Section \ref{sec:binary_split} below, illustrations of our variants are shown Figures \ref{fig:test_constraint_diagram_3cases}, \ref{fig:size_constraint_diagram}, and \ref{fig:noisy_algo_diag} below, and we highlight the algorithmic differences and key ideas the start of each respective section.

For noiseless sparsity-constrained group testing, the most relevant existing results are summarized in Table \ref{tab:sparse_algo_summary}. Our algorithm for finitely divisible items matches that of the COMP algorithm\footnote{The COMP algorithm simply labels any item in an negative test as non-defective, and all other items as defective.} in the number of tests when $\gamma = \omega(1)$ (and comes close more generally), while having much lower decoding time. Furthermore, our algorithm for size-constrained tests uses an order-optimal $O(n/\rho)$ number of tests, and has matching $O(n/\rho)$ decoding time.

For noisy non-adaptive group testing under the noise model in \eqref{eq:noise}, the most relevant existing results are summarized in Table \ref{tab:noisy_algo_summary}.  Under $\Omega(n)$-decoding time, we note that the references shown are only illustrative examples, and that several additional works also exist with $O(k \log n)$ scaling, e.g., \cite{Mal78,Sca17b,Oli20a}.  More relevant to our work is the fundamental limitation that the works attaining $O(k \log n)$ scaling only attain a quadratic or worse dependence in $k$ in the decoding time (or $\Omega(n)$).  On the other hand, GROTESQUE and SAFFRON attain $k \, {\rm poly}(\log n)$ decoding time, but fail to attain order-optimality in the number of tests.

In a distinct but related line of works, the for-all recovery guarantee (i.e., zero error probability) was considered \cite{Che09,Ind10,Ngo11,Hus19,cher20}, with typical results for the unconstrained setting requiring $O(k^2 \log n)$ tests and ${\rm poly}(k \log n)$ decoding time.  In particular:
\begin{itemize}
    \item In the finitely divisible setting, \cite{Hus19} gives a lower bound of $\Omega\big(\min\big\{n,k^{\frac{2k}{\gamma-1+k}}n^{\frac{k}{\gamma-1+k}}\big\}\big)$ and an algorithm that requires $O\big(\min\big\{n,kn^{\frac{k}{\gamma-1+k}}\big\}\big)$ tests and runs in $\text{poly}(k)+O(T)$ time in the case of $\gamma$-divisible items, and a lower bound of $\Omega\big(k\frac{n}{\rho}\big)$ and an algorithm that requires $T = O\big(k\frac{n}{\rho}\big)$ tests and runs in $\text{poly}(k)+O(T)$ time in the case of $\rho$-sized tests. 
    \item In a setting with {\em adversarial} noise, recovery guarantees were given in \cite[Thms.~3.8 and 3.9]{cher20} with a constraint on the number of false positive tests or false negative tests.  It was left open how to handle both false positives and false negatives simultaneously.
\end{itemize}
Under all variants of the group testing problem that we consider, the stronger for-all guarantee comes at the price of requiring considerably more tests.  Thus, the two types of guarantee are both of significant interest but not directly comparable, and we omit direct comparisons.

Finally, we briefly mention that studies of sublinear-time decoding are prevalent in related problems such as sparse recovery \cite{Cor06,Gil07,Ber08a,Ind11} and the heavy hitters problem \cite{Cor05a,Cor08,Lar19}.  While algorithms for such settings typically do not transfer directly to the group testing problem, we detail one relatively direct approach for the noisy setting in Appendix \ref{sec:non_binary}, and contrast it with our own.  In addition, we note that our work builds primarily on \cite{cher20,Eri20}, which in turn built on tree-based algorithms such as \cite{Cor05a,Ind11}.

\subsection{Overview of Binary Splitting Approach} \label{sec:binary_split}

Since we build directly on the fast binary splitting approach of \cite{cher20,Eri20}, we briefly summarize it here.  An illustration is given in Figure \ref{fig:binary_split}.  The items are arranged into recursively-defined groups in a sequence of levels, where Level 0 contains all items, subsequent levels recursively split the previous groups in half, and the final level contains individual items.  The shaded groups in the left part of Figure \ref{fig:binary_split} are those containing defectives, and an equivalent tree representation (with nodes corresponding to groups) is shown on the right.

At each level,\footnote{For improved efficiency, one can skip the early levels and start where there are $k$ groups of size $\frac{n}{k}$ each, up to rounding.} a suitably-chosen number of \emph{non-adaptive random tests} is performed, where items in a group are always tested together.  Whenever a group is in a negative test, the algorithm knows (in the noiseless setting) that its items must be non-defective.  Hence, when moving from one level to the next, only the sub-groups of groups in positive tests are kept under consideration.  At the final level, sufficiently many random tests are performed to identify the status of every item that has not yet been ruled out.  We refer the reader to \cite{cher20,Eri20} for further details.

\begin{figure}[!t]
    \centering
    \includegraphics[width=0.85\textwidth]{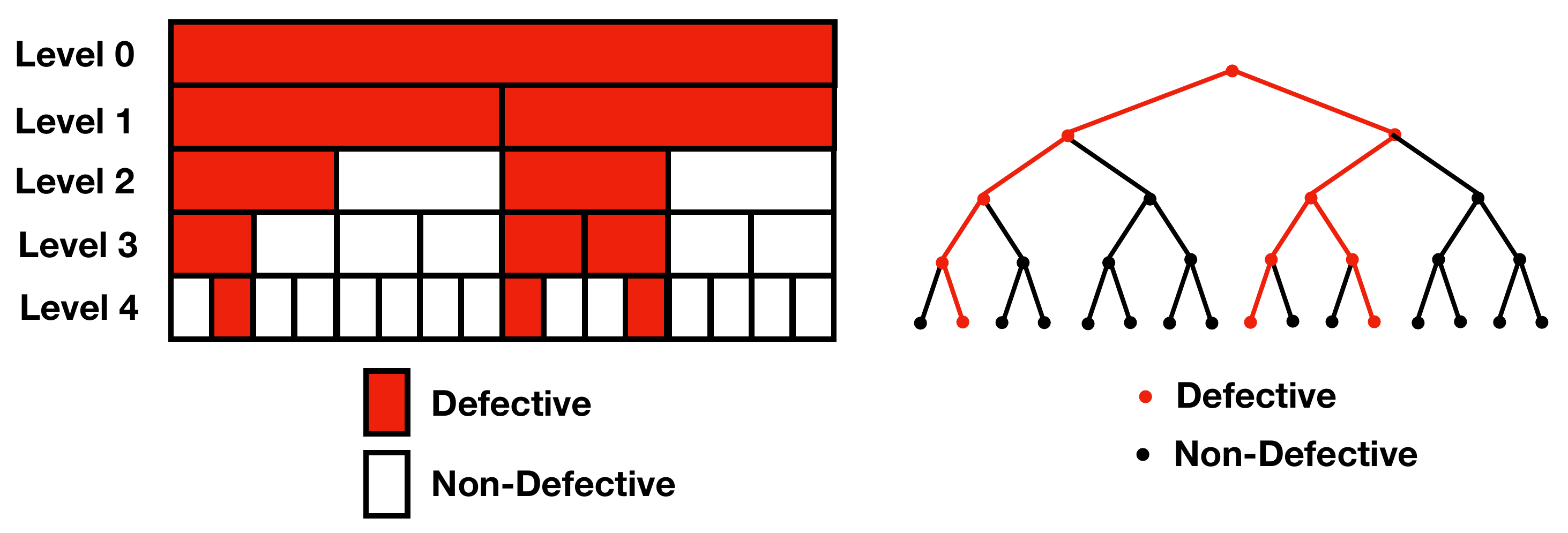}
    \caption{Illustration of the fast binary splitting approach for non-adaptive group testing proposed in \cite{cher20,Eri20}.  Here we have $n=16$ items and $k=3$ defectives; the groupings of items are shown on the left, and the tree representation (where internal nodes correspond to groups of nodes) is shown on the right.} \label{fig:binary_split}
\end{figure}

\begin{table} [t]
\centering
\begin{tabular}{|c|c|c|c|c|}
\hline
& Reference & Number of tests & Decoding time & Construction \\
\hline \hline
\multirow{5}{*}{\STAB{\rotatebox[origin=c]{90}{\makecell{$\gamma$-divis.~items~~}}}}
& Lower Bound \cite{Nel20,Oli20} & $\Omega\big(\gamma k\max\big\{k,\frac{n}{k}\big\}^{1/\gamma}\big)$ & -- & -- \\
& Gandikota {\em et al.} \cite{Ven19} & $O\big(\gamma k^2\big(\frac{n}{k^2}\big)^{1/\gamma}\big)$ & $O\big(k^2\log\big(\frac{n}{k^2}\big)\big)$ & Explicit \\
& COMP \cite{Ven19} & $\widetilde{O}(\gamma kn^{1/\gamma})$ & $\Omega(n)$ & Randomized \\
& DD \cite{Oli20} & $O\big(\gamma k\max\big\{k,\frac{n}{k}\big\}^{\frac{1}{\gamma}}\big)$ & $\Omega(n)$ & Randomized \\
& This Paper & $\widetilde{O}(\gamma kn^{1/\gamma})$ & $O(\gamma kn^{1/\gamma})$ & Randomized \\
\hline
\hline
\multirow{5}{*}{\STAB{\rotatebox[origin=c]{90}{\makecell{$\rho$-sized tests~~}}}}
& Lower Bound \cite{Ven19,Oli20} & $\Omega\big(\frac{n}{\rho}\big)$ & -- & -- \\
& Gandikota {\em et al.} \cite{Ven19} & \makecell{$O\big(\max\big\{\frac{n}{\rho}\log\rho,$\\$k^2\log\big(\frac{n}{k^2}\big)\big\}\big)$} & $O(T)$ & Explicit \\
& COMP \& DD \cite{Ven19, Oli20} & $O\big(\frac{n}{\rho}\big)$ & $\Omega(n)$ & Randomized \\
& This Paper & $O\big(\frac{n}{\rho}\big)$ & $O\big(\frac{n}{\rho}\big)$ & Randomized \\
\hline
\end{tabular}
\caption{Overview of noiseless non-adaptive sparsity-constrained group testing results under the for-each guarantee. For entries containing $\widetilde{O}(\cdot)$ notation, the results correspond to $\frac{1}{{\rm poly}(\log n)}$ error probability, but more general variants are also available. A construction is said to be explicit if its test matrix can be computed deterministically in ${\rm poly}(n)$ time; the results shown for explicit constructions additionally require $k = O(\sqrt n)$.}
\label{tab:sparse_algo_summary}
\end{table}

\begin{table} [t]
    \centering
    \begin{tabular}{|c|c|c|c|}
        \hline
        Reference & Number of tests & Decoding time & Construction \\
        \hline \hline
        Lower Bound \cite{Mal78} & $\Omega\big(k\log\frac{n}{k}\big)$ & -- & -- \\
        Inan {\em et al.} \cite{Ina19} & $O(k\log n)$ & $\Omega(n)$ & Explicit \\
        Inan {\em et al.}~(fast) \cite{Ina20} & $O(k\log n)$ & $O\big(k^3\cdot\log k+k\log n\big)$ & Explicit \\
        NCOMP \& NDD \cite{Cha14,Sca18b} & $O(k\log n)$ & $\Omega(n)$ & Randomized \\
        GROTESQUE \cite{Cai13} & $O(k\cdot\log k\cdot\log n)$ & $O\big(k(\log n+\log^2k)\big)$ & Randomized \\
        SAFFRON \cite{Lee16} & $O(k\cdot\log k\cdot\log n)$ & $O(k\cdot\log k\cdot \log n)$ & Randomized \\
        BMC \cite{Bon19a} & $O(k\log n)$ & $O(k^2\cdot \log k\cdot \log n)$ & Randomized \\
        This Paper & $O(k\log n)$ & $O\big(\big(k\log\frac{n}{k}\big)^{1+\epsilon}\big)$ & Randomized \\
        \hline
    \end{tabular}
    \caption{Overview of noisy non-adaptive group testing results under the for-each guarantee and the noise model in \eqref{eq:noise}. A construction is said to be explicit if its test matrix can be computed deterministically in ${\rm poly}(n)$ time, and in the final row, $\epsilon$ is an arbitrarily small positive constant.}
    \label{tab:noisy_algo_summary}
\end{table}

\subsection{Summary of Results}

Here we informally summarize our main results, formally stated in Theorems \ref{thm:noisy_main_theorem}, \ref{thm:gamma_main_theorem}, and \ref{thm:rho_main_theorem}.

\begin{itemize}
    \item \textbf{Finitely divisible items:} A special case of our result states that for any $\beta_n = \frac{1}{{\rm poly}(\log n)}$, there exists a non-adaptive group testing algorithm that succeeds with probability $1-O(\beta_n)$ using $\widetilde{O}\big(\gamma kn^{1/\gamma}\big)$ tests and $O\big(\gamma kn^{1/\gamma}\big)$ decoding time provided that $\gamma = \omega(1)$.  The case of finite $\gamma$ will also be handled with only slightly worse scaling laws, and we will specify the precise dependence on $\beta_n$, without resorting to $\widetilde{O}(\cdot)$ notation.
    \item \textbf{Size-constrained tests:} For any $\zeta>0$, there exists a non-adaptive group testing algorithm that succeeds with probability $1-O\big(n^{-\zeta}\big)$ using $O\big(n/\rho\big)$ tests and $O\big(n/\rho\big)$ decoding time.
    \item \textbf{Noisy setting:} For any parameters $t=O(1)$ and $\epsilon\in(1/t,1)$, there exists a non-adaptive group testing algorithm that succeeds with probability $1-O\big(\big(k\log\frac{n}{k}\big)^{1-\epsilon t}\big)$ using $O(k\log n)$ tests and $O\big(\big(k\log\frac{n}{k}\big)^{1+\epsilon}\big)$ decoding time.
\end{itemize}
We observe that in the sparsity-constrained setting, our decoding time matches the number of tests, whereas previous algorithms using the same number of tests incurred $\Omega(n)$ decoding time.  Similarly, in the noisy setting, we significantly improve on the best previous known decoding time among any algorithm using an order-optimal $O(k \log n)$ number of tests.  Specifically, \cite{Bon19a} incurred a quadratic dependence on $k$, whereas we incur a near-linear dependence.

Each of the above results comes with significant differences in the algorithms and mathematical analyses compared to the noiseless unconstrained setting handled in \cite{Eri20,cher20}.  We defer discussions of these differences to the beginning of the respective sections to follow.

While our focus is on the number of tests and decoding time, another important practical consideration is the storage required.  Naively, the algorithms attaining the above results require $\Omega(n)$ storage.  However, in Appendix \ref{sec:storage_reductions}, we discuss storage reductions via hashing, attaining identical results with sublinear storage in the size-constrained and noisy settings, and similar (but slightly weaker) results in the finitely divisible setting.


\section{Algorithm for Finitely Divisible Items}

Our algorithm (both here and in subsequent sections) resembles the non-adaptive binary splitting approach of \cite{cher20,Eri20}.  At a high level, we form large groups of items and recursively split them into smaller sub-groups, then randomly place groups into tests.  The decoder works down the resulting tree (see Figure \ref{fig:test_constraint_diagram_3cases}), eliminating groups that are believed to be defective based on the test outcomes, while recursively handling all remaining groups.

We highlight the following differences compared to the binary splitting approach \cite{cher20,Eri20}:
\begin{itemize}
    \item We use a shorter tree of height $\gamma'\leq\gamma$.  This is because a given item is placed in a single test at each level, so the assumption $\gamma = o(\log n)$ prohibits us from having $O(\log n)$ levels.  We consider $\gamma' \le \gamma$ so that the remaining budget can be used at the final level, and we later optimize $\gamma'$ to minimize the number of tests.
    \item In view of the shorter height, we use {\em non-binary} splitting; this was considered under adaptive testing in \cite{Nel20,Oli20}, and our algorithm can be viewed as a non-adaptive counterpart, in the same way that \cite{cher20,Eri20} can be viewed as a non-adaptive counterpart of Hwang's binary splitting algorithm \cite{Hwa72}.
    \item In contrast to the unconstrained setting, we cannot readily use the idea of using $N$ sequences of tests at each level while only increasing the number of tests by a factor of $N = O(1)$.  Here, such an approach turns out to be highly wasteful in terms of its use of the limited $\gamma$ budget, and we avoid it altogether.
    \item At the top level of the tree (excluding the root), we use individual testing (i.e., each node has its own test). This guarantees that no non-defective node from the second level can ``continue'' down the tree, which simplifies our analysis.
\end{itemize}

\subsection{Description of the Algorithm} \label{sec:gamma_algo_descrip}

The levels of the tree, summarized in Figure \ref{fig:test_constraint_diagram_3cases}, are indexed by $l=0,1,\dots,\gamma'$. Since testing at the root is not informative (we will always get a positive outcome), we start our testing procedure at $l=1$ (the second level of nodes in Figure \ref{fig:test_constraint_diagram_3cases}). We choose\footnote{Here and subsequently, we assume for notation convenience that $(n/k)^{1/\gamma}$ and $(n/k)^{1/\gamma'}$ are integers.  Since we focus on scaling laws, the resulting effect of rounding has no impact on our results.} $M=(n/k)^{\frac{\gamma'-1}{\gamma'}}$, $T_{\text{len}}=Ck(n/k)^{1/\gamma'}$ and $T'_{\text{len}}=\gamma'k(n/k)^{1/\gamma'}$, where $C$ is a constant.  Here the choice of $M$ is taken to match the near-optimal adaptive splitting algorithm of \cite{Nel20}, and the choices of $T_{\text{len}}$ and $T'_{\text{len}}$ are motivated by the goal of having a number of tests matching the COMP algorithm (see Table \ref{tab:sparse_algo_summary}).  Under these preceding choices, the total number of tests (excluding the last level) is given by
\begin{align}
    \underbrace{\frac{n}{M}}_{l=1}+\underbrace{\gamma'\cdot Ck\Big(\frac{n}{k}\Big)^{\frac{1}{\gamma'}}}_{l=2,\dotsc,\gamma'-2}+\underbrace{\gamma'k\Big(\frac{n}{k}\Big)^{\frac{1}{\gamma'}}}_{l=\gamma'-1}
    &=O\bigg(\gamma'k\Big(\frac{n}{k}\Big)^{\frac{1}{\gamma'}}\bigg).
\end{align} 
The overall testing procedure is described in Algorithm \ref{alg:gamma_nonadap_testing}, and the decoding procedure is described in Algorithm \ref{alg:gamma_nonadap_decoding}.  The $j$-th node at the $l$-th level is again written as $\mathcal{G}_j^{(l)}$. 

Here and subsequently, we assume that $\gamma \ge 3$.  We note that the case $\gamma = 1$ is trivial, and while $\gamma = 2$ could be handled by omitting the step at level $l = \gamma'$ containing $T''_{\rm len}$ tests, this variant is omitted for the sake of brevity.

\begin{figure}[!t]
  \centering
  \includegraphics[scale=0.4]{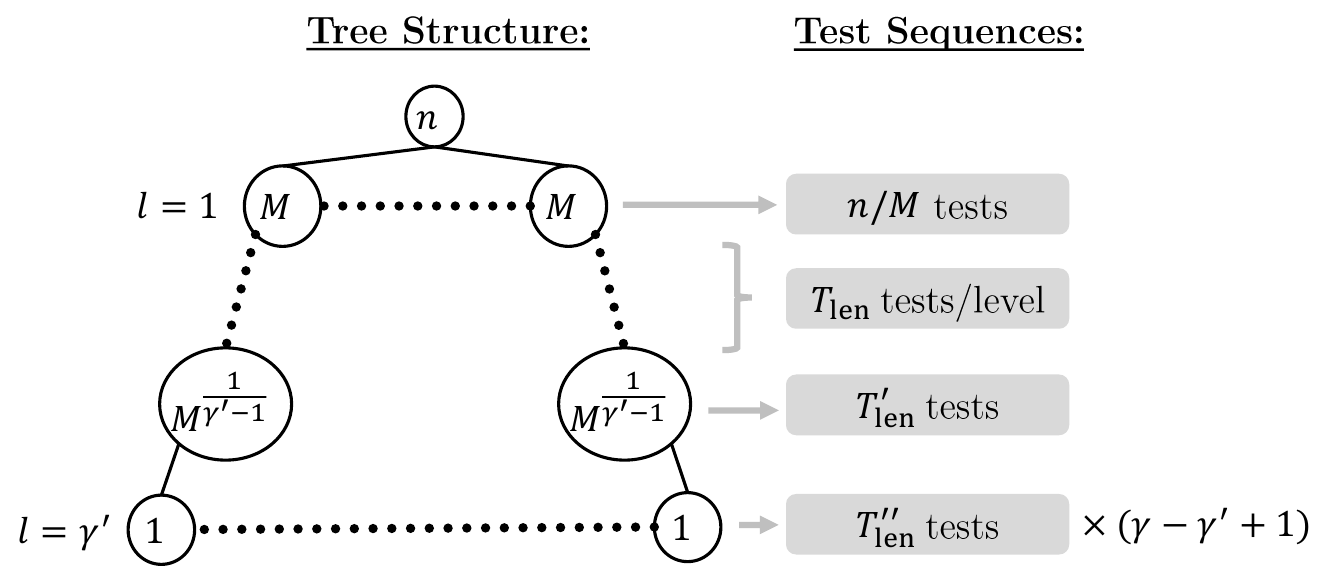}
  \caption{Tree structure of our algorithm. After the first level, the branching factor is $M^{\frac{1}{\gamma'-1}}$.} \label{fig:test_constraint_diagram_3cases}
\end{figure}

\begin{algorithm}[!t]
    \begin{algorithmic}[1]
        \REQUIRE Number of items $n$, number of defective items $k$, divisibility of each item $\gamma$, and parameters $\gamma'$, $M$, $T_{\text{len}}$, $T_{\text{len}}'$, and $T_{\text{len}}''$
        \STATE At $l=1$, test each node separately in a single test (no randomization).
        \FOR{each $l=2,3,\dots,\gamma'-1$}{
        \STATE \textbf{if} $l=\gamma'-1$ \textbf{then} form a sequence of tests of length $T'_{\text{len}}$.
        \STATE \textbf{else} form a sequence of tests of length $T_{\text{len}}$.
        \FOR{$j=1,2,\dots,\frac{n}{M}(M)^{(l-1)/(\gamma'-1)}$}{
        \STATE Place all items from $\mathcal{G}_j^{(l)}$ into a single test within the sequence just formed, chosen uniformly at random.
        }\ENDFOR
        }\ENDFOR
        \STATE For $l=\gamma'$, form $\gamma-\gamma'+1$ sequences of tests, each of length $T''_{\text{len}}$.
        \FOR{each singleton at the final level}{
        \FOR{each of the $\gamma-\gamma'+1$ sequences of tests}{
        \STATE Place the item in one of the tests in the sequence, chosen uniformly at random.
        }\ENDFOR
        }\ENDFOR
    \end{algorithmic}
    \caption{Testing procedure for $\gamma$-divisible items \label{alg:gamma_nonadap_testing}}
\end{algorithm}

\begin{algorithm}[!t]
    \begin{algorithmic}[1]
        \REQUIRE Outcomes of $T$ non-adaptive tests, number of items $n$, number of defective items $k$, divisibility of each item $\gamma$, and parameters $\gamma'$, $M$, $T_{\text{len}}$, $T_{\text{len}}'$, and $T_{\text{len}}''$
        \STATE Initialize $\mathcal{PD}^{(l_{\text{min}})}=\big\{\mathcal{G}_j^{(l_{\text{min}})}\big\}_{j=1}^{n/M}$, where $l_{\text{min}}=1$.
        \STATE Place all nodes at $l=1$ with a positive test outcome into $\mathcal{PD}^{(l_{\text{min}})}$.
        \FOR{$l=2,3,\dots,\gamma'-1$}{
        \FOR{each group $\mathcal{G}\in\mathcal{PD}^{(l)}$}{
        \STATE Check whether the single test corresponding to $\mathcal{G}$ is positive or negative.
        \STATE \textbf{if} the test is positive \textbf{then} add all $M^{1/(\gamma'-1)}$ children of $\mathcal{G}$ to $\mathcal{PD}^{(l+1)}$
        }\ENDFOR
        }\ENDFOR
        \STATE Let the estimate $\widehat{\mathcal{S}}$ of the defective set be the elements in $\mathcal{PD}^{(\gamma')}$ that are not included in any of the negative tests from the remaining $(\gamma-\gamma'+1)T_{\text{len}}''$ tests.
        \STATE Return $\widehat{\mathcal{S}}$.
    \end{algorithmic}
    \caption{Decoding procedure for $\gamma$-divisible items \label{alg:gamma_nonadap_decoding}}
\end{algorithm}

\subsection{Algorithmic Guarantees}

\begin{theorem} \label{thm:gamma_main_theorem}
Let $\mathcal{S}$ be a fixed (defective) subset of $\{1,\dots,n\}$ of cardinality $k$, and let $\gamma=o(\log n)$ (with $\gamma \ge 3$) be the maximum number of times each item can be tested, and fix $\gamma'\in\{3,\dots,\gamma\}$ and any function $\beta_n$ decaying as $n$ increases. There exist choices\footnote{Specifically, we will set $T_{\text{len}}=O\big( k(n/k)^{1/\gamma'} \big)$, $T'_{\text{len}}=\gamma'k(n/k)^{1/\gamma'}$, and $T_{\text{len}}''=k(k/\beta_n)^{\frac{1}{\gamma-\gamma'+1}}(n/k)^{\frac{1}{\gamma'(\gamma-\gamma'+1)}}$.} of $T_{\rm len}$, $T'_{\rm len}$, and $T''_{\rm len}$ such that with 
\begin{align}
    T&=O\bigg(\gamma k\max\bigg\{\Big(\frac{n}{k}\Big)^{\frac{1}{\gamma'}},
    \Big(\frac{k}{\beta_n}\Big)^{\frac{1}{\gamma-\gamma'+1}}\Big(\frac{n}{k}\Big)^{\frac{1}{\gamma'(\gamma-\gamma'+1)}}\bigg\}\bigg),
    \label{eq:gamma_main_thm_test} 
\end{align}
the preceding algorithm satisfies the following with probability at least $1-O(\beta_n)-e^{-\Omega(k)}$:
\begin{itemize}
    \item The returned estimate $\widehat{\mathcal{S}}$ equals $\mathcal{S}$;
    \item The decoding time is\footnote{In certain scaling regimes, this decoding time may be lower than the number of tests. This is because the algorithm sequentially decides which tests outcomes to observe, and does not necessarily end up observing every outcome.} $O\big(\gamma k(n/k)^{1/\gamma'}\big)$.
\end{itemize}
\end{theorem}
The proof of Theorem \ref{thm:gamma_main_theorem} is given in Appendix \ref{sec:gamma_algo_analysis}.  It consists of bounding the probabilities of non-defective nodes being ``reached'' (i.e., considered possibly defective in Line 4 of Algorithm \ref{alg:gamma_nonadap_decoding}) based on their distance to the nearest defective node.  More distant nodes have a smaller associated probability, and we can leverage this to bound the overall number of nodes visited (and hence the decoding time).  A separate analysis is also performed for the final level to show that the final estimate is correct.

\begin{figure}[!t]
  \centering
  \includegraphics[scale=0.5]{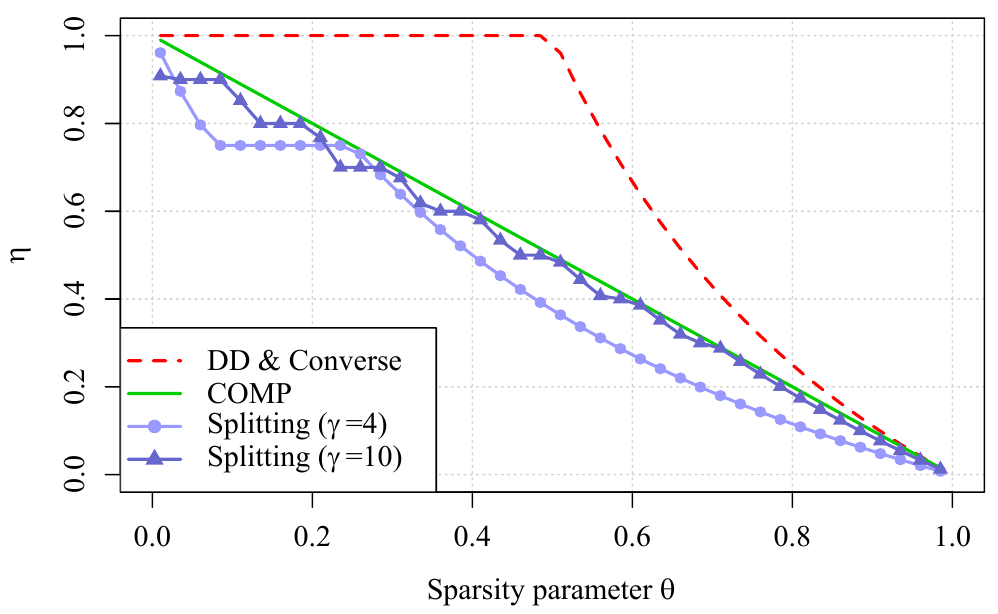}
  \caption{Plot of the asymptotic quantity $\eta$ (see \eqref{eq:eta}) against the sparsity parameter $\theta$ for the converse (i.e., the lower bound) \cite{Oli20}, the DD algorithm \cite{Oli20}, the COMP algorithm \cite{Ven19}, and our splitting algorithm (with $\gamma=4$ and $\gamma = 10$).} \label{fig:eta_plot_diff_gammas}
  \vspace*{-3ex}
\end{figure}

In order to better understand this bound on $T$, we consider $k=\Theta\big(n^{\theta}\big)$ for some $\theta\in[0,1)$, and $\beta_n = \frac{1}{{\rm poly}(\log n)}$.  These choices allow us to hide the dependence on $\beta_n$ in $\widetilde{O}(\cdot)$ notation and focus on the remaining terms. Substituting $k=\Theta\big(n^{\theta}\big)$ into \eqref{eq:gamma_main_thm_test}, we obtain
\begin{align}
    T&=\widetilde{O}\Big(\gamma k\max\Big\{n^{\frac{1-\theta}{\gamma'}},
    n^{\frac{\theta}{\gamma-\gamma'+1}+\frac{1-\theta}{\gamma'(\gamma-\gamma'+1)}}\Big\}\Big). \label{eq:regime_specific_T}
\end{align}
Momentarily ignoring the integer constraint on $\gamma'$, we obtain the optimal $\gamma'$ by solving $\frac{1-\theta}{\gamma'}=\frac{\theta}{\gamma-\gamma'+1}+\frac{1-\theta}{\gamma'(\gamma-\gamma'+1)}$, which simplifies to $\gamma'=(1-\theta)\gamma$. Substituting $\gamma'=(1-\theta)\gamma$ back into \eqref{eq:gamma_main_thm_test} gives $T=\widetilde{O}\big(\gamma kn^{1/\gamma}\big)$. In addition, by the same substitution, we obtain $O\big(\gamma kn^{1/\gamma}\big)$ decoding time. In this case, the bound on $T$ is the same as the bound for the COMP algorithm (see Table \ref{tab:sparse_algo_summary}).

When $\gamma = \omega(1)$, it is straightforward to establish that the integer constraint on $\gamma'$ does not impact the above findings. However, when $\gamma = O(1)$, we need to account for the integer constraint.  One could naively search over $\gamma'\in\{3,\dots,\gamma\}$, but in Appendix \ref{sec:convex}, we use a convexity argument to show that considering $\gamma'\in\{3,\lfloor(1-\theta)\gamma\rfloor,\lceil(1-\theta)\gamma\rceil\}$ is sufficient.

To see how our algorithm compares to optimal behavior established in \cite{Oli20} (i.e., an upper bound for the DD algorithm, and a matching algorithm-independent lower bound) and the COMP algorithm \cite{Ven19} for different values of $\gamma$, we introduce the following quantity:
\begin{align}
    \eta&=\lim_{n\rightarrow\infty}\frac{\log\big(\frac{n}{k}\big)}{\gamma\log\big(\frac{T}{\gamma k}\big)}. \label{eq:eta}
\end{align}
Observe that for any fixed value of $\eta>0$, re-arranging gives $T\sim\gamma k\big(\big(\frac{n}{k}\big)^{1/\gamma}\big)^{(1+o(1))/\eta}$. With $\eta$ defined, we compare the performance in Figure \ref{fig:eta_plot_diff_gammas}. We observe that the splitting algorithm's curve quickly gets closer to the COMP algorithm's curve even for fairly low values of $\gamma$.  On the other hand, matching the DD algorithm's curve with sublinear decoding time remains an interesting open challenge for future work.

\section{Algorithm for Size-Constrained Tests}

In the case of size-constrained tests, we again modify the tree structure (see Figure \ref{fig:size_constraint_diagram}), and the main differences from the standard noiseless algorithm \cite{cher20,Eri20} are as follows:
\begin{itemize}
    \item The first level after the root is chosen to have groups of size $\rho$, since larger groups are prohibited.  In addition, at this level with nodes of size $\rho$, we test each node individually, guaranteeing that we only ``continue'' down the tree for defective nodes at that level.
    \item We use non-binary splitting, geometrically decreasing the node size at each level until the final level with size one.  We limit the number of levels to be $O(1)$, whereas binary splitting would require $O(\log \rho)$ levels, and (at least when using a similar level-by-level test design) would increase the number of tests by an $O(\log \rho)$  factor.
    \item We do not independently place nodes into tests, since doing so would cause a positive probability of violating the $\rho$-sized test constraint.  Instead, at each level, we create a random testing sub-matrix with a column weight of exactly one, and a row weight exactly equal to to $\frac{\rho}{\text{node~size}}$. A similar doubly-constant test design was also adopted in \cite{Ven19}, but without the tree structure.
\end{itemize}
We now proceed with a more detailed description.

\subsection{Description of the Algorithm} \label{sec:rho_algo_descrip}

Our algorithm works with a tree structure (see Figure \ref{fig:size_constraint_diagram}) similar to previous sections.  The $j$-th node at the $l$-th level is again denoted by $\mathcal{G}_j^{(l)}$.  A distinction here as that the tree only has a constant depth, with the final index denoted by $C = O(1)$; hence, the splits are $\rho^{1/C}$-ary.\footnote{For notational convenience, we assume that $\rho^{1/C}$ is an integer. Since we already assumed that $\rho$ is a power of two, if $\rho=O(1)$, then it will suffice to let $C$ be that power (see Lemma \ref{lem:asymp}, in which we handle the case $\rho = O(1)$ separately). Otherwise, if $\rho = \omega(1)$, then the rounding is insignificant since $C = O(1)$.}  More importantly, there are key differences in the allocation of items to tests, which we describe as follows.  

At each level $l$, we perform $N$ independent iterations to boost the error probability, as mentioned above.  Within each iteration, we make use of a random matrix, which we write as $\mathsf{X}_l=\big[x_{ti}^{(l)}\big]\in\{0,1\}^{\text{\#tests}\times\text{\#nodes}}$ (the dependence on the iteration number is left implicit), where $\text{\#tests}=n/\rho$ and $\text{\#nodes}=\frac{n}{\rho^{1-l/C}}$. We pick $\mathsf{X}_l$ by sampling uniformly from all $\frac{n}{\rho}\times\frac{n}{\rho^{1-l/C}}$ matrices with exactly $\rho^{l/C}$ nodes per test (i.e., a row weight of $\rho^{l/C}$), and each node sampled exactly once (i.e., a column weight of one).  These choices ensure that each test contains at most $\rho$ items, as required.  The column weight of one is not strictly imposed by the testing constraints, but helps in avoiding ``bad'' events where some nodes are not tested.

With this notation in place, the testing procedure is formally described in Algorithm \ref{alg:rho_nonadap_testing}, and the decoding  procedure is described in Algorithm \ref{alg:rho_nonadap_decoding}.

\begin{figure}[!t]
  \centering
  \includegraphics[scale=0.4]{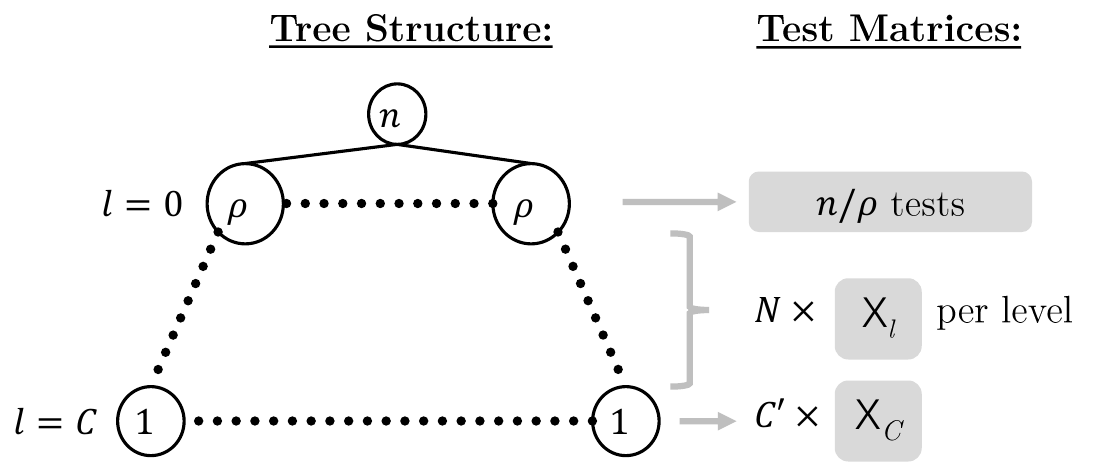}
  \caption{Tree structure of our algorithm. After the first level, the branching factor is $\rho^{1/C}$.} \vspace*{-2ex} \label{fig:size_constraint_diagram}
\end{figure}

\begin{algorithm}[!t]
    \begin{algorithmic}[1]
        \REQUIRE Number of items $n$, number of defective items $k$, and maximal test size $\rho$
        \STATE At level $l=0$ (see Figure \ref{fig:size_constraint_diagram}), perform an individual test for each node.
        \FOR{each $l=1,\dots,C-1$ (for some constant $C$ chosen later)}{
        \FOR{each iteration in $\{1,\dots,N\}$ (for some constant $N\geq1$)} {
        \STATE Pick a new $\mathsf{X}_l$ of size $\frac{n}{\rho}\times\frac{n}{\rho^{1-l/C}}$, with column weight 1 and row weight $\rho^{l/C}$.
        \FOR{each row $t$ in $\mathsf{X}_l$}{
        \STATE Conduct a single test containing the nodes $\mathcal{G}_j^{(l)}$ with $x_{tj}^{(l)}=1$.
        }\ENDFOR
        }\ENDFOR
        }\ENDFOR
        {\em Final level:}
        \FOR{each iteration in $\{1,\dots,C'\}$ (for some constant $C'$ chosen later)}{
        \STATE Pick a new $\mathsf{X}_C$ of size $\frac{n}{\rho}\times n$, with column weight 1 and row weight $\rho$.
        \FOR{each row $t$ in $\mathsf{X}_C$}{
        \STATE Conduct a single test containing the (singleton) nodes $\mathcal{G}_j^{(l)}$ with $x_{tj}^{(l)}=1$.
        }\ENDFOR
        }\ENDFOR
    \end{algorithmic}
    \caption{Testing procedure for $\rho$-sized tests \label{alg:rho_nonadap_testing}}
\end{algorithm}

\begin{algorithm}[!t]
    \begin{algorithmic}[1]
        \REQUIRE Outcomes of $T$ non-adaptive tests, number of items $n$, number of defective items $k$, and maximal test size $\rho$
        \STATE Initialize $\mathcal{PD}^{(0)}=\{\mathcal{G}_j^{(0)}\}_{j=1}^{n/\rho}$
        \FOR{each group $\mathcal{G}\in\mathcal{PD}^{(0)}$}{
        \STATE \textbf{if} the single test of $\mathcal{G}$ is positive \textbf{then} add all children of $\mathcal{G}$ to $\mathcal{PD}^{(1)}$
        }\ENDFOR
        \FOR{$l=1,\dots,C-1$}{
        \FOR{each group $\mathcal{G}\in\mathcal{PD}^{(l)}$}{
        \STATE \textbf{if} all $N$ tests of $\mathcal{G}$ are positive \textbf{then} add all children of $\mathcal{G}$ to $\mathcal{PD}^{(l+1)}$
        }\ENDFOR
        }\ENDFOR
        \STATE Let the estimate $\widehat{\mathcal{S}}$ be the set of elements in $\mathcal{PD}^{(C)}$ that are not included in any negative test at the final level.
        \STATE Return $\widehat{\mathcal{S}}$
    \end{algorithmic}
    \caption{Decoding procedure for $\rho$-sized tests \label{alg:rho_nonadap_decoding}}
\end{algorithm}

\subsection{Algorithmic Guarantees}

We are now ready to state our main result for the case of size-constrained tests.   In this case, we slightly strengthen the assumption $k = o(n)$ to $k = n^{1-\Omega(1)}$, and we slightly strengthen the assumption $\rho = o\big( \frac{n}{k} \big)$ (see the discussion following \eqref{eq:noise}) to $\rho=(n/k)^{1-\Omega(1)}$.  These additional restrictions only rule out scaling regimes that are very close to linear (e.g., $k = \frac{n}{\log n}$), and were similarly imposed in \cite{Ven19}.

\begin{theorem} \label{thm:rho_main_theorem}
Let $\mathcal{S}$ be a (defective) subset of $\{1,\dots,n\}$ of cardinality $k=O\big(n^{1-\epsilon_1}\big)$ for some $\epsilon_1\in(0,1]$ and the test size constraint be $\rho=O\big((n/k)^{1-\epsilon_2}\big)$ for some $\epsilon_2\in(0,1]$. For any $\zeta > 0$, there exist choices of $C,C',N=O(1)$ such that with $O\big(n/\rho\big)$ tests, the preceding algorithm satisfies the following with probability $1 - O\big(n^{-\zeta}\big)$:
\begin{itemize}
    \item The returned estimate $\widehat{\mathcal{S}}$ equals $\mathcal{S}$;
    \item The decoding time is $O\big(n/\rho\big)$.
\end{itemize}
\end{theorem}

The proof of Theorem \ref{thm:rho_main_theorem} is given in Appendix \ref{sec:rho_algo_analysis}, and follows similar ideas to that of Theorem \ref{thm:gamma_main_theorem} but with suitably modified details. As summarized in Table \ref{tab:sparse_algo_summary}, this is the first algorithm to attain $O\big(n/\rho\big)$ scaling in both the number of tests and the decoding time.

\section{Algorithm for the Noisy Setting} \label{sec:noisy_algo_intro}


For the unconstrained noisy setting, we revert to {\em binary} splitting (see Figure \ref{fig:noisy_algo_diag}), as was used in \cite{cher20,Eri20}, though in Appendix \ref{sec:non_binary} we also outline a non-binary approach that follows one used for the heavy hitters problem \cite{Cor08,Ind11}.  The main difference between our noisy algorithm and \cite{cher20,Eri20} is that when deciding whether a given node is defective or not, we look {\em several levels further down the tree}, instead of only considering the single test outcome of the given node.  This complicates the analysis, and leads to a small increase in the decoding time.  Additionally, in order to reduce the effective noise level, each node in the tree is placed in multiple tests, rather than just one.

\subsection{Description of the Algorithm} \label{sec:noisy_algo_descrip}

\begin{figure}[t]
    \centering
    \includegraphics[scale=0.4]{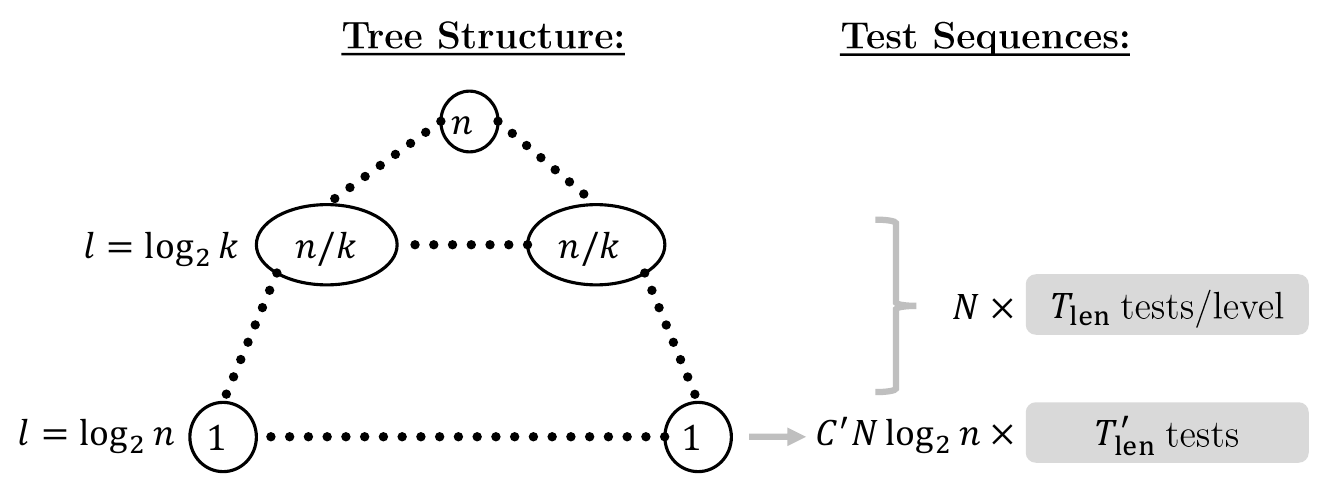}
    \caption{Tree structure of our algorithm for the noisy setting.} 
    \label{fig:noisy_algo_diag}
\end{figure}

Following \cite{cher20,Eri20}, our algorithm considers a tree representation (see Figure \ref{fig:noisy_algo_diag}), in which each node corresponds to a set of items. The levels of the tree are indexed by $l=\log_2k,\dots,\log_2n$ and the $j$-th node at the $l$-th level is denoted by $\mathcal{G}_j^{(l)}\subseteq\{1,\dots,n\}$.  At the top level we have $|\mathcal{G}_j^{(\log_2 k)}| = \frac{n}{k}$, and the sizes are subsequently halved until the final level with $|\mathcal{G}_j^{(\log_2 n)}| =1$.

The algorithm works down the tree one level at a time, keeping a list of \textit{possibly defective} ($\mathcal{PD}$) nodes, and performing tests to obtain such a list at the next level. When we perform tests at a given level, we treat each node as a ``super-item''; including a node in a test amounts to including all of the items in the corresponding node $\mathcal{G}_j^{(l)}$. In addition, for the tree illustrated in Figure \ref{fig:noisy_algo_diag}, we refer to nodes containing at least one defective item as defective nodes, to all other nodes as non-defective nodes, and to the sub-tree of defective nodes as the defective tree.

The testing is performed as follows: At each level of the tree, $N$ sequences of tests are formed, each having length $T_{\text{len}}$ (i.e., a total of $NT_{\rm len}$ tests per level).  For each node and each of the $N$ sequences, the node is placed into a single test, chosen uniformly at random among the $T_{\text{len}}$ tests.

We define the \textit{intermediate label} and \textit{final label} of a given node as follows:
\begin{itemize}
    \item The intermediate is formed via majority voting of the $N$ tests the node is included in.
    \item To obtain the final label of a given node, we look at the intermediate labels of all nodes up to $r$ levels below the given node. If there exists any length-$r$ path below the given node with more than $r/2$ positive intermediate labels, then we assign the node's final label to be positive. Otherwise, we assign it to be negative. 
\end{itemize}
According to the tree structure in Figure \ref{fig:noisy_algo_diag}, once we reach the later levels, there may be fewer than $r$ levels remaining.  To account for such cases, we simply ensure that sufficiently many tests are performed at the final level so that a length-$r$ ``path'' can be formed (here, no further branching is done, and each ``node'' is the same singleton).

With the above notation and terminology in place, the overall test design is described in Algorithm \ref{alg:noisy_nonadap_testing}, and the decoding procedure in Algorithm \ref{alg:noisy_nonadap_decoding}.


\begin{algorithm}[!t]
    \begin{algorithmic}[1]
        \REQUIRE Number of items $n$, number of defective items $k$, and parameters $N$, $C$, and $C'$
        \STATE Initialize $T_{\text{len}}=Ck$
        \FOR{each $l=\log_2k,\dots,\log_2n-1$}{
            \FOR{each iteration in $\{1,\dots,N\}$}{
                \STATE Form a sequence of tests of length $T_{\text{len}}$
                \FOR{$j=1,2,\dots,2^l$}{
                    \STATE Place all items from $\mathcal{G}_j^{(l)}$ into a single test in the sequence just formed, chosen uniformly at random.
                }\ENDFOR
            }\ENDFOR
        }\ENDFOR
        \STATE At level $l=\log_2n$, form $C'N\log_2n$ sequences of tests, each of length $T_{\text{len}}$.
        \FOR{each singleton at the final level}{
            \FOR{each of the $C'N\log_2n$ test sequences}{
                \STATE Place the singleton in one of the $T_{\text{len}}$ tests in the sequence, chosen uniformly at random.
            }\ENDFOR
        }\ENDFOR
    \end{algorithmic}
    \caption{Testing procedure for the noisy setting \label{alg:noisy_nonadap_testing}}
\end{algorithm}

\begin{algorithm}[!t]
    \begin{algorithmic}[1]
        \REQUIRE Outcomes of $T$ non-adaptive tests, number of items $n$, number of defective items $k$, and parameters $N$, $C$, $C'$, and $r$
        \STATE Initialize $\mathcal{PD}^{(l_{\text{min}})}=\big\{\mathcal{G}_j^{(l_{\text{min}})}\big\}_{j=1}^{k}$, where $l_{\text{min}}=\log_2k$.
        \FOR{$l=\log_2k,\dots,\log_2n-1$}{
            \IF{$l\leq\log_2n-r$ (i.e., there are at least $r$ levels below the node)}
            \FOR{each group $\mathcal{G}\in\mathcal{PD}^{(l)}$}{
                \STATE Evaluate the intermediate labels of all nodes $r$ levels below $\mathcal{G}$.
            }\ENDFOR
            \ELSIF{$l>\log_2n-r$ (i.e., there are fewer than $r$ levels below the node)}
            \FOR{each group $\mathcal{G}\in\mathcal{PD}^{(l)}$}{
                \STATE Evaluate the intermediate labels of all nodes all levels below $\mathcal{G}$ except the final level.
                \FOR{each node reached at the final level}{
                    \STATE Iterate through the $C'N\log_2n$ test outcomes in batches of size $N$: Conduct a majority vote for each batch to obtain an intermediate label for the node. \label{op:sub_alg_final_lvl}
                }\ENDFOR
                \STATE Use intermediate labels from each node in the final level to make up paths of length $r$ (see Section \ref{sec:noisy_algo_descrip}).
            }\ENDFOR
            \ENDIF
            \FOR{each group $\mathcal{G}\in\mathcal{PD}^{(l)}$}{
                \STATE If $\exists$ a path with more than $r/2$ positive intermediate labels, then assign $\mathcal{G}$'s final label to be positive. Otherwise, assign $\mathcal{G}$'s final label to be negative.
                \STATE If the final label of $\mathcal{G}$ is positive, then add both children of $\mathcal{G}$ to $\mathcal{PD}^{(l+1)}$.
            }\ENDFOR
        }\ENDFOR
        \STATE At the final level, for each node (singleton), repeat step \ref{op:sub_alg_final_lvl} to obtain $C'\log_2n$ intermediate labels for the node, and conduct a majority vote for the node's intermediate labels to obtain its final label. 
        \STATE Return $\widehat{\mathcal{S}}$ containing the elements of singletons in $\mathcal{PD}^{(\log_2n)}$ with a positive final label.
    \end{algorithmic}
    \caption{Decoding procedure for the noisy setting \label{alg:noisy_nonadap_decoding}}
\end{algorithm}

\subsection{Algorithmic Guarantees} \label{sec:noisy_algo_guaran}

\begin{theorem} \label{thm:noisy_main_theorem}
    Let $\mathcal{S}$ be a (defective) subset of $\{1,\dots,n\}$ of cardinality $k=o(n)$. For any constants $\epsilon > 0$ and $t > 0$ satisfying $\epsilon t>1$, there exist choices of $C,C',N = O(1)$ and $r = O(\log k + \log\log n)$ such that with $O\big(k\log n\big)$ tests, the preceding algorithm satisfies the following with probability at least $1-O\big(\big(k\log\frac{n}{k}\big)^{1-\epsilon t}\big)$:
    \begin{itemize}
        \item The returned estimate $\widehat{\mathcal{S}}$ equals $\mathcal{S}$;
        \item The decoding time is $O\big(\big(k\log\frac{n}{k}\big)^{1+\epsilon}\big)$.
    \end{itemize}
\end{theorem}
The proof of Theorem \ref{thm:noisy_main_theorem} is given in Appendix \ref{sec:noisy_algo_analysis}.  The main distinction compared to the noiseless proofs is that we need to bound the probabilities of intermediate labels (used in Lines 5 and 8 of Algorithm \ref{alg:noisy_nonadap_decoding}) and final labels (computed in Line 13) being wrong, to ensure that correct decisions are made at each level.  The $N$ independent repetitions at each level play the role of reducing the former, and the independence of tests across levels helps to tightly characterize the latter.

\section{Conclusion} \label{sec:conclusion}

We have provided fast splitting algorithms for sparsity-constrained and noisy group testing, maintaining the near-optimal number of tests provided by earlier works while also attaining a matching or near-matching decoding time.  Possible directions for future research include (i) in the finitely divisible setting, match the number of tests used by the DD algorithm (see Table \ref{tab:sparse_algo_summary}) with sublinear decoding time, and (ii) in the noisy setting, further reduce the $\big( k \log \frac{n}{k} \big)^{1+\epsilon}$ runtime, ideally bringing it all the way down to $O(k \log n)$.  


\appendix

\section*{\huge Appendix}

\section{Proof of Theorem \ref{thm:gamma_main_theorem} (Finitely Divisible Items)} \label{sec:gamma_algo_analysis}
Throughout the analysis, the defective set $\mathcal{S}$ is fixed but otherwise arbitrary, and we condition on fixed placements of the defective items into tests (and hence, fixed test outcomes and a fixed defective tree). The test placements of the non-defective items are independent of those of the defective items, and our analysis will hold regardless of which particular tests the defectives were placed in. The defective test placements are written as $\mathcal{T}_{\mathcal{S}}$, and we write $\mathbb{P}[\cdot \,|\,  \mathcal{T}_{\mathcal{S}}]$ to denote the conditioning.

We proceed with three lemmas that follow analogous steps to \cite{Eri20}.  At level $l=1$, the probability of a non-defective node being placed in a positive test is zero, because each node is placed in its own individual test. As for levels $l \in \{2,\dotsc,\gamma'-2\}$, we proceed with the following simple lemma.

\begin{lemma} \label{lem:gamma_prob_of_nondef_node_being_in_pos_test} {\textup{(Probabilities of Non-Defectives Being in Positive Tests)}}
Under the above test design, the following holds at any given level $l=2,\dots,\gamma'-2$: Conditioned on any defective test placements $\mathcal{T}_{\mathcal{S}}$, any given non-defective node at level $l$ has probability at most $(1/C)(n/k)^{-1/\gamma'}$ of being placed in a positive test.
\end{lemma} 

\begin{proof}
Since there are $k$ defective items, at most $k$ nodes at a given level can be defective. Hence, since each node is placed in a single test, at most $k$ tests out of the $Ck(n/k)^{1/\gamma'}$ tests at the given level can be positive. Since the test placements are independent and uniform, it follows that for any non-defective node, the probability of being in a positive test is at most $k/T_{\text{len}}=k/\big(Ck(n/k)^{1/\gamma'}\big)=(1/C)(n/k)^{-1/\gamma'}$.
\end{proof}

In view of this lemma, when starting at any non-defective child of any defective node, we can view any further branches down the non-defective sub-tree as ``continuing'' (i.e., the $M^{1/(\gamma'-1)}$ children are marked as possibility defective) with probability at most $(1/C)(n/k)^{-1/\gamma'}$, in particular implying Lemma \ref{lem:gamma_prob_of_nondef_node_at_dist_away} below.  
Before stating the lemma, we introduce some terminology that well help us make more concise statements:
\begin{itemize}
    \item We say that a node is \emph{reached} if all of its ancestors are placed in positive tests, so the node will be considered possibly defective.  This is in contrast to nodes that are not reached (by the decoding algorithm) because one of their ancestors is found to be non-defective.
    \item For any non-defective node, we define its \emph{distance to the defective tree} as the smallest number of edges that needs to be traversed to reach a defective node (e.g., $\Delta = 1$ for a non-defective child of a defective node).
\end{itemize}

\begin{lemma} \label{lem:gamma_prob_of_nondef_node_at_dist_away} \textup{(Probability of Reaching a Non-Defective Node)}
    Under the setup of Lemma \ref{lem:gamma_prob_of_nondef_node_being_in_pos_test}, any given non-defective node at distance $\Delta$ from the defective tree is reached with probability at most $(1/C)^{\Delta-1}(n/k)^{(1-\Delta)/\gamma'}$.
\end{lemma} 

We will use the preceding lemmas to control the quantity $N_{\text{total}}$, defined to be the total number of non-defective nodes that are \textit{reached}---in the sense of Lemma \ref{lem:gamma_prob_of_nondef_node_at_dist_away}---among levels $l\in\{2,\dots,\gamma'-1\}$. It will be useful to upper bound $N_{\text{total}}$ for the purpose of controlling the overall decoding time and the number of items considered at the final level.

\subsection{Bounding $N_{\text{total}}$} 

We first present a lemma bounding the average of $N_{\text{total}}$.

\begin{lemma} \label{lem:N_total_average_bound_gamma} \textup{(Bounding $N_{\text{total}}$ on Average)}
For any parameters $C>1$ and $\gamma'>1$, and any defective test placements $\mathcal{T}_{\mathcal{S}}$, under the choice $M=(n/k)^{\frac{\gamma'-1}{\gamma'}}$, we have
\begin{align}
    \E[N_{\textup{total}}|\mathcal{T}_{\mathcal{S}}]=O\bigg(\gamma'k\Big(\frac{n}{k}\Big)^{\frac{1}{\gamma'}}\bigg).
\end{align}
\end{lemma} 

\begin{proof}
At level $l=1$, we use $n/M$ tests for individual nodes. This results in correct identification of the non-defective nodes, guaranteeing that they will not ``continue'' to branch. Hence, at level $l=1$, we trivially upper bound the number of non-defective nodes by $n/M$.

For the remaining levels $l=2,\dots,\gamma'-1$, all splits are $\big(M^{1/(\gamma'-1)}\big)$-ary, and each defective node can have at most $M^{\Delta/(\gamma'-1)}$ descendants at distance $\Delta$. Since there are at most $\gamma'k$ defective nodes in total among levels $l=1,\dots,\gamma'-1$, it follows that there are at most $\gamma'kM^{\frac{\Delta}{\gamma'-1}}$ non-defective nodes at distance $\Delta$ from defective nodes starting at those levels.  Furthermore, we established in Lemma \ref{lem:gamma_prob_of_nondef_node_at_dist_away} that a distance of $\Delta$ gives a probability of at most $\big(\frac{1}{C}\big)^{\Delta-1}\big(\frac{n}{k}\big)^{(1-\Delta)/\gamma'}$ of being reached. This gives
\begin{align}
    \E[N_{\textup{total}}|\mathcal{T}_{\mathcal{S}}]
    &\leq\sum_{\Delta=1}^{\gamma'}\gamma'kM^{\frac{\Delta}{\gamma'-1}}\Big(\frac{1}{C}\Big)^{\Delta-1}\Big(\frac{n}{k}\Big)^{\frac{1-\Delta}{\gamma'}}
    +\frac{n}{M} \\
    &=\gamma'kM^{\frac{1}{\gamma'-1}}\sum_{\Delta=1}^{\gamma'}M^{\frac{\Delta-1}{\gamma'-1}}\Big(\frac{1}{C}\Big)^{\Delta-1}\Big(\frac{n}{k}\Big)^{\frac{1-\Delta}{\gamma'}}+\frac{n}{M} \\
    &\stackrel{(a)}{\le}\gamma'kM^{\frac{1}{\gamma'-1}}\frac{1}{1-M^{\frac{1}{\gamma'-1}}\big(\frac{1}{C}\big)\big(\frac{n}{k}\big)^{-1/\gamma'}}+\frac{n}{M} \\
    &\stackrel{(b)}{=}\gamma'k\Big(\frac{n}{k}\Big)^{\frac{1}{\gamma'}}\frac{1}{1-1/C}+k\Big(\frac{n}{k}\Big)^{\frac{1}{\gamma'}},
\end{align}
where (a) applies the geometric series formula (increasing the upper limit of the sum from $\gamma'$ to $\infty$), and (b) follows by substituting $M=(n/k)^{\frac{\gamma'-1}{\gamma'}}$.
\end{proof}

We now wish to move from a characterization of the average to a high-probability characterization.  At this point, we depart somewhat further from the analysis of \cite{Eri20}, which is based on branching process theory, and appears to yield suboptimal results in the case that the tree's branching factor scales as $\omega(1)$.

We introduce the following definition, in which we refer to a \textit{full} $m$-ary tree as a tree where every \textit{internal} node has exactly $m$ children.
 
\begin{lemma}\label{lem:cat_num} \textup{\cite[Prop.~3.1]{AVAL08} (Fuss-Catalan Numbers)} 
For natural integers $m,n\geq2$, the order-$m$ Fuss-Catalan number
\begin{align}
    \textup{Cat}_m^{n}&=\frac{1}{(m-1)n+1}{mn\choose n}
    \leq{mn\choose n}
    \leq(em)^n,
\end{align}
is the number of full $m$-ary trees with exactly $n$ internal nodes.
\end{lemma}

We note that the Catalan numbers also played an important role in the analysis of the unconstrained setting it \cite{cher20}, but were used in a rather different manner that we were unable to extend to obtain a result comparable to Theorem \ref{thm:gamma_main_theorem}.  In the proof of the following lemma, these are used in a counting argument in order to establish the sub-exponential behavior of the random variable $N_{\rm total}$.

\begin{lemma} \label{lem:N_total_high_prob_bound_gamma} \textup{(High Probability Bound on $N_{\text{total}}$)}
For any parameters $C\geq e^2$ and $\gamma'>1$, and any defective test placements $\mathcal{T}_{\mathcal{S}}$, under the choice $M=(n/k)^{\frac{\gamma'-1}{\gamma'}}$, we have $N_{\textup{total}}=O\big(\gamma'k(n/k)^{1/\gamma'}\big)$ with probability $1-e^{-\Omega(\gamma'k)}$.
\end{lemma}

\begin{proof}
Consider a single non-defective sub-tree following a defective node, and let $N_b$ be the number of nodes in the sub-tree such that itself and all its ancestors only appear in positive tests (i.e., the number of nodes that lead to further branching).  We have
\begin{align}
    \mathbb{P}[N_b=n_b]&\leq\mathbb{P}[\exists\text{ a full $M^{1/(\gamma'-1)}$-tree reached with $n_b$ internal nodes}] \\
    &\stackrel{(a)}{\leq}(\text{\#full $M^{1/(\gamma'-1)}$-trees with $n_b$ internal nodes})\cdot\bigg(\frac{1}{C}\Big(\frac{n}{k}\Big)^{-\frac{1}{\gamma'}}\bigg)^{n_b} \\
    &\stackrel{(b)}{\leq}\big(eM^{1/(\gamma'-1)}\big)^{n_b}\bigg(\frac{1}{C}\Big(\frac{n}{k}\Big)^{-\frac{1}{\gamma'}}\bigg)^{n_b} \\
    &\stackrel{(c)}{=}\Big(\frac{e}{C}\Big)^{n_b}
    \stackrel{(d)}{\leq}e^{-n_b},
\end{align}
where (a) applies Lemma \ref{lem:gamma_prob_of_nondef_node_being_in_pos_test} and the union bound, (b) applies Lemma \ref{lem:cat_num}, (c) is obtained by substituting $M=(n/k)^{\frac{\gamma'-1}{\gamma'}}$ and simplifying, and (d) holds since $C\geq e^2$. This implies that $N_b$ is a sub-exponential random variable. Since we have at most $(\gamma'-1)k$ defective nodes in levels $l=1,\dots,\gamma'-1$, we are adding together $O(\gamma'k)$ independent copies of such random variables (each corresponding to a different non-defective sub-tree following a defective node).\footnote{We do not consider the non-defective nodes at level $l=1$, because they are guaranteed to be identified correctly as a result of individual testing of nodes.} Letting $N_b^{(i)}$ denote the $i$-th copy, we can apply a standard concentration bound for sums of independent sub-exponential random variables \cite[Prop. 5.16]{Ver12} to obtain 
\begin{align}
    \mathbb{P}\big[N_b^{(1)}+\dots+N_b^{(O(\gamma'k))}&\geq\E[N_b^{(1)}+\dots+N_b^{(O(\gamma'k))}]+t|\mathcal{T}_\mathcal{S}\big]\leq\exp\bigg(\Omega\Big(\min\Big\{\frac{t^2}{\gamma'k},t\Big\}\Big)\bigg).
\end{align}
Setting $t=\Theta(\gamma'k)$, we get
\begin{align}
    \mathbb{P}[N_b^{(1)}+\dots+N_b^{(O(\gamma'k))}\geq\E[N_b^{(1)}+\dots+N_b^{(O(\gamma'k))}]+\Theta(\gamma'k)|\mathcal{T}_\mathcal{S}]\leq e^{-\Omega(\gamma'k)}. \label{eq:N_conc}
\end{align}
Recall that each $N_b^{(i)}$ only counts ``internal'' nodes, whereas $N_{\rm total}$ also counts leaves, so passing from the former to the latter requires multiplying by the branching factor $M^{1/(\gamma'-1)}=(n/k)^{1/\gamma'}$.   Multiplying on both sides inside the probability in \eqref{eq:N_conc} accordingly, we obtain
\begin{align}
    \mathbb{P}\bigg[N_{\text{total}}\geq\E[N_{\text{total}}]+\Theta\bigg(\gamma'k\Big(\frac{n}{k}\Big)^{1/\gamma'}\bigg)\Big|\mathcal{T}_\mathcal{S}\bigg]\leq e^{-\Omega(\gamma'k)}. \label{eq:concentration_bound}
\end{align}
Substituting $\E[N_{\text{total}}]=O\big(\gamma'k(n/k)^{1/\gamma'}\big)$ (see Lemma \ref{lem:N_total_average_bound_gamma}) into \eqref{eq:concentration_bound}, we obtain the desired result.
\end{proof}

We now briefly consider level $l=\gamma'-1$, which uses $T'_{\rm len} = \gamma'k(n/k)^{1/\gamma'}$ tests (see Figure \ref{fig:test_constraint_diagram_3cases}).  Since $|\mathcal{PD}^{(\gamma'-1)}|\leq N_{\text{total}}+k$ holds trivially, Lemma \ref{lem:N_total_high_prob_bound_gamma} implies that $|\mathcal{PD}^{(\gamma'-1)}|=O\big(\gamma'k(n/k)^{1/\gamma'}\big)$ with probability $1-e^{-\Omega(\gamma' k)}$. Using the same argument as Lemma \ref{lem:gamma_prob_of_nondef_node_being_in_pos_test}, the probability of a non-defective node being in a positive test at level $l=\gamma'-1$ is at most $k/T'_{\text{len}}=(1/\gamma')(n/k)^{-1/\gamma'}$. Hence, conditioned on $|\mathcal{PD}^{(\gamma'-1)}|=O\big(\gamma'k(n/k)^{1/\gamma'}\big)$, the number of non-defective nodes placed in a positive test is stochastically dominated by
\begin{align}
    \text{Binomial}\bigg( O\Big( \gamma'k\Big(\frac{n}{k}\Big)^{1/\gamma'} \Big),\frac{1}{\gamma'}\Big(\frac{n}{k}\Big)^{-1/\gamma'}\bigg).
\end{align}
By a multiplicative form of Chernoff bound, the number of such non-defective nodes in $\mathcal{PD}^{(\gamma'-1)}$ is $O(k)$ with probability at least $1-e^{-\Omega(k)}$. Since the branching factor is $(n/k)^{1/\gamma'}$, it follows that the number of non-defective nodes in $\mathcal{PD}^{(\gamma')}$ behaves as $O(k(n/k)^{1/\gamma'})$.

\subsection{Analysis of the Final Level} \label{sec:gamma_final_lvl_analysis}

Recall that at the final level, we perform $\gamma-\gamma'+1$ independent sequences of tests of length $T_{\text{len}}''$, with each item being randomly placed in one of these $T_{\text{len}}''$ tests. Conditioned on the high probability event that $|\mathcal{PD}^{(\gamma')}|=O(k(n/k)^{1/\gamma'})$, we study the required $T_{\text{len}}''$ for a vanishing error probability. Specifically, we upper bound the error probability by $O(\beta_n)$ for some decaying function $\beta_n\rightarrow0$ as $n\rightarrow\infty$. 

For a given non-defective item and a given sequence of $T_{\text{len}}''$ tests, the probability of colliding with any defective item is at most $k/T_{\text{len}}''$ by the same argument as Lemma \ref{lem:gamma_prob_of_nondef_node_being_in_pos_test}. Due to the $\gamma-\gamma'+1$ independent repetitions, the probability of a given non-defective item appearing only in positive tests is at most $(k/T_{\text{len}}'')^{\gamma-\gamma'+1}$. By a union bound over $O(k(n/k)^{1/\gamma'})$ non-defective items at the final level, we find that the estimate $\widehat{\mathcal{S}}$ differs from $\mathcal{S}$ with (conditional) probability $O\big(k(n/k)^{1/\gamma'}(k/T_{\text{len}}'')^{\gamma-\gamma'+1}\big)$. The error probability is thus upper bounded by $O(\beta_n)$ provided that
\begin{align}
    &k\Big(\frac{n}{k}\Big)^{\frac{1}{\gamma'}}\Big(\frac{k}{T_{\text{len}}''}\Big)^{\gamma-\gamma'+1}\leq\beta_n \\
    \iff &T_{\text{len}}''\geq k\Big(\frac{k}{\beta_n}\Big)^{\frac{1}{\gamma-\gamma'+1}}\Big(\frac{n}{k}\Big)^{\frac{1}{\gamma'(\gamma-\gamma'+1)}}.
\end{align}
Hence, we set $T_{\text{len}}''=k(k/\beta_n)^{\frac{1}{\gamma-\gamma'+1}}(n/k)^{\frac{1}{\gamma'(\gamma-\gamma'+1)}}$.

\subsection{Number of Tests, Error Probability, and Decoding Time}

\begin{itemize}
    \item \textbf{Number of tests:} For $l=1,\dots,\gamma'-1$, we used a total of $n/M+ C(\gamma'-3)k(n/k)^{1/\gamma'}+\gamma'k(n/k)^{1/\gamma'}$ tests, which scales as $O\big(\gamma'k(n/k)^{1/\gamma'}\big)$ by substituting $M=(n/k)^{\frac{\gamma'-1}{\gamma'}}$ and $C = O(1)$. For the final level, we used $(\gamma-\gamma'+1)T_{\text{len}}''=O\big(\gamma k(k/\beta_n)^{\frac{1}{\gamma-\gamma'+1}}(n/k)^{\frac{1}{\gamma'(\gamma-\gamma'+1)}}\big)$ tests, due to the fact that  $T_{\text{len}}''=k(k/\beta_n)^{\frac{1}{\gamma-\gamma'+1}}(n/k)^{\frac{1}{\gamma'(\gamma-\gamma'+1)}}$. Combining these, we obtain
    \begin{align}
        T&=O\bigg(\gamma k\max\bigg\{\Big(\frac{n}{k}\Big)^{\frac{1}{\gamma'}}
        ,\Big(\frac{k}{\beta_n}\Big)^{\frac{1}{\gamma-\gamma'+1}}\Big(\frac{n}{k}\Big)^{\frac{1}{\gamma'(\gamma-\gamma'+1)}}\bigg\}\bigg). \label{eq:splitting_algo__upper_bound}
    \end{align}
    \item \textbf{Error probability:} The concentration bound on $N_{\text{total}}$ (see Lemma \ref{lem:N_total_high_prob_bound_gamma}) holds with probability $1-e^{-\Omega(\gamma'k)}$, and at level $l=\gamma'-1$, we incur $e^{-\Omega(k)}$ error probability. Furthermore, the final stage incurs $O(\beta_n)$ error (conditional) probability. In total, we incur $\beta_n+e^{-\Omega(\gamma'k)}+e^{-\Omega(k)}=O(\beta_n)+e^{-\Omega(k)}$ error probability.
    \item \textbf{Decoding time:} We claim that conditioned on the high-probability events above (in particular, $N_{\text{total}}=O\big(\gamma'k(n/k)^{1/\gamma'}\big)$), the decoding time is $O\big(\gamma k(n/k)^{1/\gamma'}\big)$. Since we consider the word-RAM model, it takes constant time to check whether each defective node or non-defective node is in a positive or negative test. First considering the levels $l=2,\dotsc,\gamma'-1$, we reached $N_{\text{total}}=O\big(\gamma'k(n/k)^{1/\gamma'}\big)$ non-defective nodes and $O(\gamma'k)$ defective nodes, which leads to a total of $O\big(\gamma'k(n/k)^{1/\gamma'}\big)$ decoding time.  At level $l=1$, we iterate through $\frac{n}{M} = O\big(k(n/k)^{1/\gamma'}\big)$ nodes, and at the final level $l=\gamma'$, for each of the $O\big(k(n/k)^{1/\gamma'}\big)$ relevant leaf nodes, we perform $\gamma-\gamma'+1=O(\gamma)$ checks of tests for a total time of $O\big(\gamma k(n/k)^{1/\gamma'}\big)$. Combining these terms, we deduce the desired claim.
\end{itemize}

\subsection{Note on Optimizing $\gamma'$} \label{sec:convex}

We note that the function $f(\gamma') = \max\big\{\frac{1-\theta}{\gamma'},\frac{\theta}{\gamma-\gamma'+1}+\frac{1-\theta}{\gamma'(\gamma-\gamma'+1)}\big\}$ is convex on $[3,\gamma]$; this is easily proved by computing the second derivative of each term in $\max\{.,.\}$. 
Since a convex function is monotone on either side of its minimum (in this case $(1-\theta)\gamma$), it follows that the optimal choice of $\gamma'$ is given by
\begin{align}
    \gamma'&=\argmin\limits_{\gamma'\in\{3,\dots,\gamma\}}\Big(\gamma k\max\Big\{n^{\frac{1-\theta}{\gamma'}},
    n^{\frac{\theta}{\gamma-\gamma'+1}+\frac{1-\theta}{\gamma'(\gamma-\gamma'+1)}}\Big\}\Big) \\
    &=
    \begin{cases}
    3 &\text{if $(1-\theta)\gamma<3$} \\
    \argmin\limits_{\gamma'\in\{\lfloor(1-\theta)\gamma\rfloor,\lceil(1-\theta)\gamma\rceil\}}\bigg(\max\bigg\{\frac{1-\theta}{\gamma'},\frac{\theta}{\gamma-\gamma'+1}+\frac{1-\theta}{\gamma'(\gamma-\gamma'+1)}\bigg\}\bigg) &\text{otherwise.} \label{eq:gamma'}
    \end{cases}
\end{align}
That is, we can simply evaluate the objective for three values of $\gamma'$, rather than all values.

\section{Proof of Theorem \ref{thm:rho_main_theorem} (Size-Constrained Tests)} \label{sec:rho_algo_analysis}

We start at level $l=0$ (see Figure \ref{fig:size_constraint_diagram}), where we note that the probability of a non-defective node being placed in a positive test is zero because each node is placed in its own individual test.  For subsequent levels, we proceed with the following lemma.

\begin{lemma} \label{lem:rho_prob_of_nondef_node_being_in_pos_test} {\textup{(Probabilities of Non-Defectives Being in Positive Tests)}}
Under the above test design, for any given level $l=1,\dots,C$ and any given iteration indexed by $\{1,\dotsc,N\}$, each non-defective node has probability at most $k\rho/n$ of being placed in a positive test.
\end{lemma} 

\begin{proof}
At any given iteration of level $l$, the probability that a non-defective node $u$ collides (i.e., is in the same test) with a given defective node $v$ is
\begin{align}
    \frac{\text{\#matrices with $u$ \& $v$ in test $1$}}{\text{\#matrices with $v$ in test $1$}}
    &\stackrel{(a)}{=}\frac{{\frac{n}{\rho^{1-l/C}}-2\choose\rho^{l/C}-2}\prod_{i=1}^{n/\rho-1}{\frac{n}{\rho^{1-l/C}}-i\rho^{l/C}\choose\rho^{l/C}}}
    {{\frac{n}{\rho^{1-l/C}}-1\choose\rho^{l/C}-1}\prod_{i=1}^{n/\rho-1}{\frac{n}{\rho^{1-l/C}}-i\rho^{l/C}\choose\rho^{l/C}}} \\
    &=\frac{{\frac{n}{\rho^{1-l/C}}-2\choose\rho^{l/C}-2}}{{\frac{n}{\rho^{1-l/C}}-1\choose\rho^{l/C}-1}}
    \stackrel{(b)}{=}\frac{\rho^{l/C}-1}{\frac{n}{\rho^{1-l/C}}-1} \\
    &=\frac{\rho}{n}\bigg(\frac{\rho^{l/C}-1}{\rho^{l/C}-\rho/n}\bigg)
    \stackrel{(c)}{\leq}\frac{\rho}{n},
\end{align}
where:
\begin{itemize}
    \item (a) follows by considering the rows of the matrix $\mathsf{X}_l$ (of size $\frac{n}{\rho} \times \frac{n}{\rho^{1-l/C}}$, column weight one, and row weight $\rho^{l/C}$) sequentially to count the number of possible matrices. For the numerator, we start with the first row, where $u$ and $v$ collide.  The number of ways to fill this row (i.e., assigning items to this test) is the first term in the numerator. For the remaining $n/\rho-1$ rows, in any particular order, the number of ways to fill those rows (while maintaining column weights of one) is represented by the second product term. The same analysis is then repeated for the denominator.
    \item (b) follows by expanding the binomial coefficient in terms of factorials, and then simplifying.
    \item (c) follows from the fact that $\rho/n\leq1$.
\end{itemize}
Since there are at most $k$ defective nodes, by the union bound, we find that the probability that a non-defective node collides with any defective node is at most $k\rho/n$.
\end{proof}

The following technical lemma will also be used on several occasions.

\begin{lemma} \label{lem:asymp}
    For any $k$ and $\rho$ satisfying $k=O\big(n^{1-\epsilon_1}\big)$ for some $\epsilon_1\in(0,1]$ and $\rho=O\big((n/k)^{1-\epsilon_2}\big)$ for some $\epsilon_2 \in (0,1]$, we have the following:
    \begin{itemize}
        \item For sufficiently large $C$, we have $\frac{k\rho^{1/C}}{n/\rho} = n^{-\Omega(1)}$;
        \item For any $\zeta_1 > 0$ , we have for sufficiently large $C$ and $N$ that $\rho^{1/C}\big(\frac{k\rho}{n}\big)^N = O(n^{-\zeta_1})$.
    \end{itemize}
    In addition, if $\rho = O(1)$, then the same holds true for any fixed $C>0$, only requiring $N$ to be sufficiently large in the second part.
\end{lemma}
\begin{proof}
    For the first part, we write
    \begin{align}
        \frac{n/\rho}{k\rho^{1/C}}&=\frac{n/k}{\rho^{1+1/C}}
        \stackrel{(a)}{=}\Omega\bigg(\Big(\frac{n}{k}\Big)^{\epsilon_2-\frac{1-\epsilon_2}{C}}\bigg)
        \stackrel{(b)}{=}\Omega\big(n^{\epsilon_1(\epsilon_2-\frac{1-\epsilon_2}{C})}\big), \label{eq:k*rho^(1/C)<<n/rho_proof}
    \end{align}
    where (a) is by substituting $\rho=O\big((n/k)^{1-\epsilon_2}\big)$ and simplifying, and (b) is by substituting $k=O\big(n^{1-\epsilon_1}\big)$ and simplifying. Note that the power is positive for sufficiently large $C$.

    For the second part, we write
    \begin{align}
        \rho^{1/C}\Big(\frac{k\rho}{n}\Big)^{N}
        &\stackrel{(a)}{=} O\bigg(\Big(\frac{n}{k}\Big)^{\frac{1-\epsilon_2}{C}-\epsilon_2N}\bigg)
        \stackrel{(b)}{=}O\big(n^{\epsilon_1(\frac{1-\epsilon_2}{C}-\epsilon_2N)}\big),
    \end{align}
    where (a) is by substituting $\rho=O\big((n/k)^{1-\epsilon_2}\big)$ and simplifying, and (b) is by substituting $k=O\big(n^{1-\epsilon_1}\big)$ and simplifying.  Note that the power can be made arbitrarily negative by choosing $N$ and $C$ sufficiently large.
    
    For the final part regarding $\rho = O(1)$, we simply note that the two claims reduce to (i) $\frac{k}{n} = n^{-\Omega(1)}$, and (ii) $\big(\frac{k}{n}\big)^{N} = O(n^{-\zeta_1})$ for sufficiently large $N$.  Both of these are true since $k = O(n^{1-\epsilon_1})$.
\end{proof}

We will show that throughout the course of the algorithm, for levels $l=1,\dots,C$, the size of the possibly defective set $\mathcal{PD}^{(l)}$ remains at $O\big(k\rho^{1/C}\big)$ with high probability.  We show this using an induction argument.  

\subsection{Analysis of Levels $l=1,\dots,C-1$}

For the base case $l=1$, we start by looking at the preceding level $l=0$. Each node at level $l=0$ is allocated to an individual test, which implies that all nodes in $l=0$ are identified correctly. Hence, only the children of the defective nodes in $l=0$ are ``explored'' further in $l=1$. Since the number of defective nodes in $l=0$ is at most $k$ and each node has $\rho^{1/C}$ children, we have $|\mathcal{PD}^{(1)}|\leq k\rho^{1/C}$. 

Consider a non-defective node indexed by $i$ at a given level $l > 1$ having $k'\leq k$ defective nodes, and let $A_i$ be the indicator random variable of that non-defective node colliding with at least one defective node in all of its $N$ repetitions. The dependence of these quantities on $l$ is left implicit. We condition on all of the test placements performed at the earlier levels, writing $\mathbb{E}_l[\cdot]$ for the conditional expectation. By the inductive hypothesis, we have $|\mathcal{PD}^{(l)}|=O\big(k\rho^{1/C}\big)$.

\begin{lemma} \label{lem:average_bound_rho_case}
Under the preceding setup and definitions, if $|\mathcal{PD}^{(l)}|=O\big(k\rho^{1/C}\big)$, then we have
\begin{align}
    \E_l\Big[\sum_iA_i\Big]=O\bigg(k\rho^{1/C}\cdot\Big(\frac{k\rho}{n}\Big)^{N}\bigg).
\end{align}
\end{lemma}

\begin{proof}
From Lemma \ref{lem:rho_prob_of_nondef_node_being_in_pos_test}, we know that a given non-defective item $i$ has a probability at most $k\rho/n$ of being placed in a positive test. Since we used $N$ independent test design matrices $\mathsf{X}_l$ to assign $i$ to $N$ tests, we have $\mathbb{P}_l[A_i]\leq(k\rho/n)^{N}$. Hence, we have
\begin{align}
    \E_l\Big[\sum_iA_i\Big]
    &=\sum_i\E_l[A_i]
    =\sum_i\mathbb{P}_l[A_i = 1]
    \leq\sum_i\Big(\frac{k\rho}{n}\Big)^{N}
    =O\bigg(k\rho^{1/C}\cdot\Big(\frac{k\rho}{n}\Big)^{N}\bigg),
\end{align}
where we used the linearity of expectation and the fact that $|\mathcal{PD}^{(l)}|=O\big(k\rho^{1/C}\big)$.
\end{proof}

\begin{lemma}
For any constant $\zeta_1>0$, there exist choices of $C$ and $N$ such that the following holds: Conditioned on the $l$-th level having $|\mathcal{PD}^{(l)}|=O\big(k\rho^{1/C}\big)$, the same is true at the $(l+1)$-th level with probability $1-O\big(n^{-\zeta_1}\big)$.
\end{lemma}

\begin{proof}
Among the possibly defective nodes at the $l$-th level, at most $k$ are defective, amounting to at most $k\rho^{1/C}$ children at the next level. Furthermore, by Lemma \ref{lem:average_bound_rho_case} and Markov's inequality, at most $k$ non-defective nodes are marked as possibly defective, with probability at least
\begin{align}
    1-O\bigg(\rho^{1/C}\Big(\frac{k\rho}{n}\Big)^{N}\bigg) = 1 - O(n^{-\zeta_1}), 
\end{align}
where the equality holds for any $\zeta_1 > 0$ by suitable choices of $C$ and $N$ (see Lemma \ref{lem:asymp}).
Thus, this also amounts to at most $k\rho^{1/C}$ additional children at the next level. Summing these together, we have $|\mathcal{PD}^{(l+1)}| \le 2 k\rho^{1/C}$, with probability at least $1-O\big(n^{-\zeta_1}\big)$.
\end{proof}

By induction, for any given level $l$, we have $|\mathcal{PD}^{(l)}|=O\big(k\rho^{1/C}\big)$ with conditional probability at least $1-O\big(n^{-\zeta_1}\big)$. Taking a union bound over all $C$ levels (with $C = O(1)$), the same follows for all levels simultaneously with probability at least $1-O\big(n^{-\zeta_1}\big)$.

\subsection{Analysis of the Final Level}

Recall that at the final level, we perform $C'n/\rho$ tests. We study the error probability conditioned on the high-probability event $|\mathcal{PD}^{(C)}|=O\big(k\rho^{1/C}\big)$.

For a given non-defective item in a single iteration of the $C'$ independent iterations of tests, by Lemma \ref{lem:rho_prob_of_nondef_node_being_in_pos_test}, the probability of appearing in a positive test is at most $k\rho/n$. Since the non-defective item participates in $C'$ independent tests, the probability of it appearing only in positive tests is $(k\rho/n)^{C'}$. By a union bound over the $O\big(k\rho^{1/C}\big)$ non-defective singletons at the final level, the error probability is upper bounded by
\begin{align}
    O\bigg(k\rho^{1/C}\Big(\frac{k\rho}{n}\Big)^{C'}\bigg) = O(n^{-\zeta_2}),
\end{align}
where the equality holds for any $\zeta_2 > 0$ and suitably-chosen $C$ and $C'$ due to Lemma \ref{lem:asymp} (with $C'$ replacing $N$).


\subsection{Number of Tests, Error Probability, and Decoding Time}

\begin{itemize}
    \item \textbf{Number of tests:} We used $CNn/\rho$ tests in the first $C$ levels and $C'n/\rho$ tests in the final level, which sums up to $CNn/\rho+C'n/\rho=O(n/\rho)$.
    \item \textbf{Error probability:} For each level $l$, we have $|\mathcal{PD}^{(l)}|=O\big(k\rho^{1/C}\big)$ with probability $1-O\big(n^{-\zeta_1}\big)$. Furthermore, the final level incurs $O\big(n^{-\zeta_2}\big)$ error probability. This gives us a total error probability of $O\big(n^{-\zeta_1}+n^{-\zeta_2}\big)=O\big(n^{-\zeta}\big)$, where $\zeta=\min\{\zeta_1,\zeta_2\}$.  Since we allowed $\zeta_1$ and $\zeta_2$ to be arbitrarily large, the same holds for $\zeta$.
    \item \textbf{Decoding time:} The decoding time is dominated by the test outcome checks in our decoding procedure. For the first level $l=0$, we have $|\mathcal{PD}^{(0)}|=n/\rho$, which coincides with the total number of test outcome checks. For the remaining $C-1$ levels $l\in\{1,\dots,C-1\}$, we considered a total of $O\big(k\rho^{1/C}\big)$ possibly defective nodes w.h.p.,\footnote{Here and subsequently, we write {\em with high probability} (w.h.p.) to mean holding under the high-probability events used in proving that the algorithm succeeds.} and for each possibly defective item, we conducted $N$ test outcome checks. This gives us total number of $O\big(k\rho^{1/C}\big)$ test outcome checks. At the final level, for each of the $O\big(k\rho^{1/C}\big)$ relevant leaf nodes, we perform $C'$ test outcome checks for a total time of $O\big(k\rho^{1/C}\big)$. Summing these gives $O(n/\rho)$, since $O\big(k\rho^{1/C}\big)=o(n/\rho)$ for a sufficiently large $C$ (refer to \eqref{eq:k*rho^(1/C)<<n/rho_proof}). Since it takes $O(1)$ time to check whether each node is in a positive or negative test, we get a total decoding time of $O(n/\rho)$.
\end{itemize}

\section{Proof of Theorem \ref{thm:noisy_main_theorem} (Noisy Setting)} \label{sec:noisy_algo_analysis}

The outline of the analysis is as follows:
\begin{itemize}
    \item We first consider levels $l=\log_2k,\dots,\log_2n-1$, and bound the probability that any node among three kinds---non-defective nodes at level $l_{\text{min}} = \log_2 k$, defective nodes, and non-defective child nodes of defective nodes---are identified wrongly. Note that we do not have to consider other nodes, because if none of the nodes of these three kinds are identified wrongly, then the algorithm would not explore any of the other nodes when decoding.
    \item Conditioned on the correct identification of nodes of these three kinds, we consider the final level $l=\log_2n$ and provide a bound for its error probability.
\end{itemize}

\subsection{Analysis of Levels $l=\log_2k,\dots,\log_2n-1$}

We consider defective and non-defective nodes separately.

\textbf{Defective nodes:} Recall the notions of intermediate labels and final labels from Section \ref{sec:noisy_algo_descrip}.  Let $p_{\text{int}}^{(\text{d})}$ (respectively, $p_{\text{final}}^{(\text{d})}$) be the probability that the intermediate label (respectively, final label) of a given defective node is flipped from a one to a zero. Note that these may vary from node to node, but we will give upper bounds that hold uniformly.

For a given defective node, there are only two possible situations for each test it is in: A positive outcome due to no flip, or a negative test outcome due to a $1\rightarrow0$ flip. Hence, the number of negative tests that a given defective node participates in (i.e., the outcome is flipped) is distributed as $\text{Binomial}(N,p)$. By the majority voting of $N$ test outcomes at a given level, $p_{\text{int}}^{(\text{d})}$ is upper bounded by the probability that a given defective node participates in at least $N/2$ negative tests. Applying Hoeffding's inequality, we obtain
\begin{align}
    p_{\text{int}}^{(\text{d})}&\leq\exp\bigg(-2N\Big(\frac{1}{2}-p\Big)^2\bigg).
\end{align}
At this point, we introduce the variable $t$ appearing in the theorem statement. Since $\exp\big(-2N(1/2-p)^2\big)\leq \frac{2^{-2t}}{4}\Leftrightarrow N\geq \frac{2t\log2+\log4}{2(1/2-p)^2}$, we find that choosing $N\geq \frac{2t\log2+\log4}{2(1/2-p)^2}$ ensures that
\begin{align}
    p_{\text{int}}^{(\text{d})}\leq\frac{2^{-2t}}{4}. \label{eq:rho_int_def_upperbound}
\end{align}

For the case that $l\leq\log_2n-r$, we consider the length-$r$ paths below the defective node. The defective node will be labeled as negative if all $2^r$ paths below it have at least $r/2$ negative intermediate labels. The probability of this event is upper bounded by the probability that one particular \textit{defective} path (i.e., every node along the path is defective) has at least $r/2$ negative intermediate labels, which is at most
\begin{align}
    {r\choose r/2}\big(p_{\text{int}}^{(\text{d})}\big)^{r/2}\leq\big(4p_{\text{int}}^{(\text{d})}\big)^{r/2}, \label{eq:p_final_bound}
\end{align}
where the left hand side (LHS) is by the union bound, and the right hand side (RHS) is by ${r\choose r/2}\leq2^r$.
This gives $p_{\text{final}}^{(\text{d})}\leq\big(4p_{\text{int}}^{(\text{d})}\big)^{r/2}$, and substituting \eqref{eq:rho_int_def_upperbound} gives $p_{\text{final}}^{(\text{d})}\leq2^{-tr}$. 

For the case that $l>\log_2n-r$ (i.e., there are less than $r$ levels below the given node), the probability of the (single) defective path having at least $r/2$ negative intermediate labels remains unchanged, and hence, the preceding bound $p_{\text{final}}^{(\text{d})}\leq2^{-tr}$ still holds. Note that this step requires $C'\log_2n\geq r$ in order to have enough intermediate labels per node in the final level to ``pad'' paths of length less than $r$ (see Section \ref{sec:noisy_algo_descrip}), and we will later set $C'$ and $r$ to ensure  this.

\textbf{Non-defective nodes:} Let $p_{\text{int}}^{(\text{nd})}$ (respectively, $p_{\text{final}}^{(\text{nd})}$) be the probability that the intermediate label (respectively, final label) of a given non-defective node is flipped from a zero to a one. Again, these may vary from node to node, but we will give upper bounds that hold uniformly. For a given non-defective node, there are four possible situations for each test: A negative outcome with no flip (i.e., no defectives), a negative outcome due to a $1\rightarrow0$ flip (i.e., at least one defective), a positive outcome with no flip (i.e., at least one defective), and a positive outcome due to a $0\rightarrow1$ flip (i.e., no defectives).

Focusing on one test sequence of length $T_{\text{len}} = Ck$ for now, let $A$ be the event that a given non-defective node participates in a positive test, and let $B$ be the event that the given node's test contains no defective item. We have
\begin{align}
    \mathbb{P}[A]
    &=\mathbb{P}[B]\cdot\mathbb{P}[A|B]+\mathbb{P}[\neg B]\cdot\mathbb{P}[A|\neg B] \\
    &\stackrel{(a)}{\leq}\mathbb{P}[B]\cdot p+\frac{1}{C}(1-p) \\
    &\leq p+\frac{1}{C}, \label{eq:upperbound_of_P[A]}
\end{align}
where (a) holds since the probability of being in the same test as a given defective node is $1/T_{\rm len} = 1/(Ck)$, and thus the union bound over $k$ defective nodes gives $\mathbb{P}[\neg B] \le 1/C$.

Equation \eqref{eq:upperbound_of_P[A]} implies that for a given non-defective node, the number of positive tests that it participates in (out of $N$ tests in total) is stochastically dominated by $\text{Binomial}(N,p+1/C)$. Recalling that $p_{\text{int}}^{(\text{nd})}$ is the probability that a given non-defective node participates in at least $N/2$ positive tests, Hoeffding's inequality gives
\begin{align}
    p_{\text{int}}^{(\text{nd})}&\leq\exp\bigg(-2N\Big(\frac{1}{2}-p-\frac{1}{C}\Big)^2\bigg), \label{eq:rho_int_nondef_upperbound0}
\end{align}
where we require $1/2-p-1/C>0\Leftrightarrow C>2/(1-2p)$. Hence, we set $C=\lceil2/(1-2p)\rceil+1$. Since $\exp\big(-2N(1/2-p-1/C)^2\big)\leq \frac{2^{-2t}}{16} \Leftrightarrow N\geq \frac{2t\log2+\log16}{2(1/2-p-1/C)^2}$, we find that choosing $N\geq \frac{2t\log2+\log16}{2(1/2-p-1/C)^2}$ ensures that
\begin{align}
    p_{\text{int}}^{(\text{nd})}\leq\frac{2^{-2t}}{16}. \label{eq:rho_int_nondef_upperbound}
\end{align}

For the case that $l\leq\log_2n-r$, we look at the length-$r$ path below the non-defective node. The non-defective node will be labeled as positive if any of the $2^r$ paths below it has at least $r/2$ positive intermediate labels. By a union bound over all $2^r$ paths, this probability is upper bounded as follows, similar to \eqref{eq:p_final_bound}:
\begin{align}
    2^r{r\choose r/2}\big(p_{\text{int}}^{(\text{nd})}\big)^{r/2}
    \leq2^r\big(4p_{\text{int}}^{(\text{nd})}\big)^{r/2}
    \leq\big(16p_{\text{int}}^{(\text{nd})}\big)^{r/2}.
\end{align}
This gives $p_{\text{final}}^{(\text{nd})}\leq\big(16p_{\text{int}}^{(\text{nd})}\big)^{r/2}$, and substituting \eqref{eq:rho_int_nondef_upperbound} gives $p_{\text{final}}^{(\text{nd})}\leq2^{-tr}$.

Similarly to the defective nodes handled above, the case that $l>\log_2n-r$ follows essentially unchanged; while the above analysis has an additional union bound over $2^r$ paths, the number of paths when $l > \log_2n - r$ only gets smaller.  Hence, the preceding bound on $p_{\text{final}}^{(\text{nd})}\leq2^{-tr}$ also holds in this case.


\textbf{Combining the defective and non-defective cases:} Taking the more stringent requirement on $N$ in the above two cases, we set
\begin{align}
    N&=\bigg\lceil\frac{2t\log2+\log16}{2(1/2-p-1/C)^2}\bigg\rceil, \label{eq:Nchoice}
\end{align}
and we observe that regardless of the defectivity of a given node, the probability of the node's final label being wrong is at most $2^{-tr}$. 

Next, we upper bound the probability that any node among three groups---non-defective nodes at level $l_{\text{min}}$, defective nodes, and child nodes of defective nodes---is identified wrongly.  Note that if all such nodes are identified correctly, then the branching is only ever continued for defective nodes, and it follows that at most $2k$ nodes remain at the final level (analyzed below).

Since there are $\log_2(n/k)$ levels and $k$ defectives, the number of non-defective children nodes of defective nodes is at most $k\log_2(n/k)$, and the number of non-defective nodes at level $l_{\text{min}}$ is at most $k$. Summing these up, we have at most $2k\log_2(n/k)+k$ nodes. By taking the union bound over all $2k\log_2(n/k)+k$ nodes, the probability of making an error in identifying any node in the mentioned three groups is at most $2^{-tr}(2k\log_2(n/k)+k)$. This can be upper bounded by a given target value $\beta_n$ (approaching zero as $n\rightarrow\infty$), provided that
\begin{align}
    2^{-tr}\bigg(2k\log_2\Big(\frac{n}{k}\Big)+k\bigg)\leq\beta_n, \label{eq:beta_n_eq}
\end{align}
which rearranges to give
\begin{align}
    r&\geq\frac{1}{t}\log_2\bigg(\frac{2k}{\beta_n}\log_2\Big(\frac{n}{k}\Big)+\frac{k}{\beta_n}\bigg).
\end{align}
By choosing
\begin{align}
    r=\bigg\lceil\frac{1}{t}\log_2\bigg(\frac{3k}{\beta_n}\log_2\Big(\frac{n}{k}\Big)\bigg)\bigg\rceil,
\end{align}
we deduce that the probability of any wrong decision is upper bounded by $\beta_n$.

\subsection{Analysis of the Final Level}


Recall from the analyses of \eqref{eq:rho_int_def_upperbound} and \eqref{eq:rho_int_nondef_upperbound} that given our choice of $N$ in \eqref{eq:Nchoice}, regardless of the defectivity of a given node, the probability of a wrong intermediate label---let us call this $p_{\text{int}}$---is at most $2^{-2t}/4$. To get the final label of each node (singleton), we conduct a majority voting of $C'\log_2n$ intermediate labels. Hence, a given node is labeled wrongly when it has at least $(C'\log_2n)/2$ wrong intermediate labels. This gives the following upper bound on the probability of a wrong final label, denoted by $p_{\text{final}}$:
\begin{align}
    p_{\text{final}}
    &\leq{C'\log_2n\choose (C'\log_2n)/2}\big(p_{\text{int}}\big)^{(C'\log_2n)/2}
    \stackrel{(a)}{\leq}\big(4p_{\text{int}}\big)^{(C'\log_2n)/2}
    \stackrel{(b)}{\leq}2^{-tC'\log_2n},
\end{align}
where (a) uses ${x \choose x/2} \le 2^x$, and (b) uses $p_{\text{int}}\leq2^{-2t}/4$. Taking the union bound over all $n$ nodes at the final level, we obtain
\begin{align}
    n\big(2^{-tC'\log_2n}\big)
    &=n\big(n^{-tC'}\big)
    =O(n^{1-tC'}),
\end{align}
which approaches zero as $n\rightarrow\infty$ as long as $tC' > 1$. Note that while we have shown that all $n$ nodes (singletons) at the final level would be correctly identified if their final labels were to be computed, only at most $2k$ of these will actually be used by the algorithm, in accordance with the above analysis.


\subsection{Number of Tests, Error Probability, and Decoding Time}

For convenience, we restate all the values that we have assigned in our analysis above:
\begin{align}
    C&=\bigg\lceil\frac{2}{1-2p}\bigg\rceil+1=O(1) \\
    N&=\bigg\lceil\frac{2t\log2+\log16}{2(1/2-p-1/C)^2}\bigg\rceil=O(t) \\
    r&=\bigg\lceil\frac{1}{t}\log_2\bigg(\frac{3k}{\beta_n}\log_2\Big(\frac{n}{k}\Big)\bigg)\bigg\rceil
    =O\bigg(\frac{1}{t}\log\Big(\frac{k\log(n/k)}{\beta_n}\Big)\bigg), \label{eq:r_formula_general}
\end{align}
where $p\in(0,1/2)$ is the noise level. Now, we choose $t=O(1)$ and $\beta_n=\big(k\log_2(n/k)\big)^{1-\epsilon t}$, for some constant $\epsilon\in(1/t,1)$. Substituting $\beta_n=\big(k\log_2(n/k)\big)^{1-\epsilon t}$ into \eqref{eq:r_formula_general} gives
\begin{align}
    r&=\bigg\lceil\frac{1}{t}\log_2\bigg(\frac{3k\log_2(n/k)}{\big(k\log_2(n/k)\big)^{1-\epsilon t}}\bigg)\bigg\rceil
    =\bigg\lceil\frac{1}{t}\log_2\bigg(3\Big(k\log_2\Big(\frac{n}{k}\Big)\Big)^{\epsilon t}\bigg)\bigg\rceil. \label{eq:r_formula_specific}
\end{align}
Recall that we require $C' \log_2n \geq r$, or equivalently $C'\geq r/\log_2n$. Substituting \eqref{eq:r_formula_specific} into $C'\geq r/\log_2n$, we find that we require
\begin{align}
    C'&\geq\frac{\big\lceil\frac{1}{t}\log_2\big(3\big(k\log_2\big(\frac{n}{k}\big)\big)^{\epsilon t}\big)\big\rceil}{\log_2n}, \label{eq:C'_cond}
\end{align}
Since $\epsilon$ is constant, we can choose $C'=O(1)$ that is large enough to satisfy \eqref{eq:C'_cond}. With our choices of $C,C',N,t=O(1)$ and $\beta_n=\Theta\big((k\log n)^{1-\epsilon t}\big)$, we obtain the following:

\begin{itemize}
    \item \textbf{Number of tests:} We used $CNk$ tests per level for $l=\log_2k,\dots,\log_2n-1$. At the final level $l=\log_2n$, we used $CC'Nk\log_2n$ tests. Summing these together gives
    \begin{align}
        T&\leq CNk\log_2\Big(\frac{n}{k}\Big)+CC'Nk\log_2n\stackrel{(a)}{=}O(k\log n), \label{eq:final_test_num}
    \end{align}
    where (a) follows by substituting $C,C',N=O(1)$ and simplifying.
    \item \textbf{Error probability:} Combining the error probabilities from all levels, we have a total error probability of at most
    \begin{align}
        \beta_n+O\big(n^{1-tC'}\big)=O\bigg(\Big(k\log\Big(\frac{n}{k}\Big)\Big)^{1-\epsilon t}\bigg),
    \end{align}
    by substituting $\beta_n=\big(k\log_2(n/k)\big)^{1-\epsilon t}$ and choosing $C'$ sufficiently large.
    \item \textbf{Decoding time:} To characterize the decoding time, we consider the number of test outcome checks made throughout the course of the algorithm. For $l=\log_2k,\dots,\log_2n-1$, w.h.p., we involved $O\big(k\log(n/k)\big)$ nodes in total. For each node involved, we checked at most $\sum_{i=1}^r2^{i}=O\big(2^r\big) \stackrel{\eqref{eq:r_formula_general}}{=} O\big(\big(\frac{k\log(n/k)}{\beta_n}\big)^{1/t}\big)$ intermediate labels of other nodes to decide the final label of the given node. For each these nodes being checked, we checked $N=O(t)$ test outcomes to determine the intermediate label. Therefore, the decoding time for these levels is
    \begin{align}
        O\bigg(k\log\Big(\frac{n}{k}\Big)\cdot\Big(\frac{k\log(n/k)}{\beta_n}\Big)^{1/t}\cdot t\bigg), \label{eq:decoding_time_except_final_level}
    \end{align}
    At the final level $l=\log_2n$, we have already shown that w.h.p., at most $2k$ nodes remain possibly defective. For each such node, we checked $C'\log_2n$ intermediate labels to decide the final label of the given node. To decide each intermediate label, we checked $N=O(t)$ test outcomes. Therefore, the decoding time at this level is $O(2k\cdot C'\log n\cdot t)$. Summing this with \eqref{eq:decoding_time_except_final_level} gives us the total decoding time of 
    \begin{align}
        O\bigg(k\log\Big(\frac{n}{k}\Big)\cdot\Big(\frac{k\log(n/k)}{\beta_n}\Big)^{1/t}\cdot t\bigg)+O(2k\cdot C'\log n\cdot t)
        =O\bigg(\Big(k\log\frac{n}{k}\Big)^{1+\epsilon}\bigg), \label{eq:final_decoding_time}
    \end{align}
    by substituting $C',t=O(1)$ and $\beta_n=\big(k\log_2(n/k)\big)^{1-\epsilon t}$, and noting that the $O(k \log n)$ term is dominated by $O\big(\big(k\log\frac{n}{k}\big)^{1+\epsilon}\big)$ regardless of the scaling of $k$.
\end{itemize}

\section{Non-Binary Trees in the Noisy Setting} \label{sec:non_binary}

\subsection{Unconstrained Noisy Setting}

Our algorithm for the noisy setting in Section \ref{sec:noisy_algo_intro} is based on binary splitting, and combats noise by both (i) performing independent repetitions at each level, and (ii) classifying a given node by exploring levels further down the tree.  Here we discuss an alternative approach based on non-binary splitting, which attains similar results using only the former of these.\footnote{This approach was pointed out by an anonymous reviewer of an earlier version of this paper.}  Despite this, we believe that there is value in also showing that binary splitting suffices, and that our technique of exploring further down the tree may be of independent interest.

The non-binary approach we consider in this section is based on the analysis of the heavy hitters problem in \cite[Sec.~B.2]{Lar19}, which in turn builds on \cite{Cor08}.  Instead of forming a binary tree as in Figure \ref{fig:noisy_algo_diag}, consider forming a $b$-ary tree for some value of $b$ to be chosen later.  Hence, the depth of the tree is $O\big( \frac{\log n}{\log b} \big)$.  

At each level, instead of using $O(1)$ independent repetitions (as was done in Algorithm \ref{alg:noisy_nonadap_testing}), we use $O(\log b)$ repetitions.  Since there are $O\big( \frac{\log n}{\log b} \big)$ levels, and each repetition contains $O(k)$ tests, the total number of tests is $O(k \log n)$.   In addition, by a similar analysis to that of $p_{\text{int}}^{(\text{d})}$ and $p_{\text{int}}^{(\text{nd})}$ in Appendix \ref{sec:noisy_algo_analysis}, each majority vote over these repetitions succeeds with probability at least $1 - \frac{1}{{\rm poly}(b)}$, where the polynomial has arbitrarily high degree.

When all such majority votes are correct, the algorithm only visits $O(kb)$ nodes, and thus, if $b = (k \log n)^{\epsilon}$, the probability of any wrong decision can be made to decay as $\frac{1}{{\rm poly}(k \log n)}$.  While the list size at the final level increases from $O(k)$ (in our binary splitting approach) to $O(bk)$, the final level can still be analyzed in the same way as Appendix \ref{sec:noisy_algo_analysis}, and the total decoding time is $O(kb \log n) = O\big( (k \log n)^{1+\epsilon} )$.  This is equivalent to the decoding time $O\big( \big(k \log \frac{n}{k} \big)^{1+\epsilon} )$ given in Theorem \ref{thm:noisy_main_theorem}, since if $k$ is large enough for $\log\frac{n}{k}$ to significantly differ from $\log n$, then the logarithmic factor can be factored into the $k^{\epsilon}$ term anyway.

\subsection{Noisy Setting with Size-Constrained Tests}

At first glance, it may appear to be difficult to combine our techniques for the size-constrained and noisy settings, since the latter is based on searching $\omega(1)$ levels down the tree, whereas the former uses a tree with depth $O(1)$.  However, even in \cite{Ven19} where the computation time is $\Omega(n)$, moving to the noisy setting increases the number of tests from $O\big(\frac{n}{\rho}\big)$ to $O\big(\frac{n}{\rho} \log n\big)$.  We can incur a similar increase by increasing our tree depth from $O(1)$ to $O(\log n)$, and this added depth permits us to combat noise in the same way as the unconstrained setting.  For the sake of brevity, we omit the details.

\section{Storage Reductions via Hashing} \label{sec:storage_reductions}

For all of our algorithms considered, the storage comprises of storing the assignments of nodes to tests, storing the possibly defective set $\mathcal{PD}$, and storing the test outcomes.  We observe that since every tree that we consider has a final level containing $n$ nodes, storing the test assignments at that level alone requires $\Omega(n)$ storage, meaning that the standard versions of our algorithms do not have sublinear storage.


In order to reduce the storage, we can make modifications to each algorithm in a similar manner to \cite{Eri20}:  Instead of directly storing the test outcomes of every node, we interpret the node-to-test mappings at each level (except for one-to-one mappings) as hash functions. Since the high storage comes from explicitly storing the corresponding test outcomes of nodes, the key to reducing the overall storage is to use lower-storage hash families. 

The reduced storage comes at the expense of reduced independence between different hash values. Fortunately, this drawback has a negligible effect on the guarantees of our algorithm under the noisy setting and size-constrained setting, as the proofs of Theorems \ref{thm:noisy_main_theorem} and \ref{thm:rho_main_theorem} only require pairwise independence or weaker. However, the effect is more significant for our algorithm under the finitely divisible items constraint, as our proof of Theorem \ref{thm:gamma_main_theorem} uses full independence.  In the following, we briefly describe suitable properties and choices for the hash families, and how they affect the algorithmic guarantees.  We let $\mathsf{T}_{\text{hash}}$ and $\mathsf{S}_{\text{hash}}$ respectively denote the evaluation time for one hash value and the number of bits of storage required for one hash function.

\textbf{Finitely divisible items:} Consider using an $O(\gamma)$-wise independent hash family to generate a hash function, with $\mathsf{T}_{\text{hash}}=O(\gamma)$ and $\mathsf{S}_{\text{hash}}=O(\gamma\log n)$ (e.g., see \cite[Section 3.1]{Eri20}). Since the analysis in Appendix \ref{sec:gamma_algo_analysis} requires full independence, a different analysis is required for the algorithmic guarantees. 

To address this, we note that two distinct analyses were given in \cite{Eri20}, with fully independent hashes attaining the stronger result, and limited-independence hashes reducing the storage but increasing the error probability. The latter of these in fact extends to the finitely divisible setting significantly more easily than the former does, so we simply state the corresponding result and omit the proof: For any function $\beta_n$ decaying as $n$ increases, using
\begin{align}
    T&=O\bigg(\gamma k\max\bigg\{\Big(\frac{n}{k}\Big)^{\frac{1}{\gamma'}},
    \Big(\frac{k}{\beta_n}\Big)^{\frac{1}{\gamma-\gamma'+1}}\Big(\frac{n}{k}\Big)^{\frac{1}{\gamma'(\gamma-\gamma'+1)}}\bigg\}\bigg)
\end{align}
tests, the algorithm has $O\big(\mathsf{T}_{\textup{hash}}\gamma k(n/k)^{1/\gamma'}\big)=O\big(\gamma^2k(n/k)^{1/\gamma'}\big)$ runtime, requires a storage of $O\big(k(n/k)^{1/\gamma'}\log n+\mathsf{S}_{\textup{hash}}\gamma+T\big)=O\big(\big(k(n/k)^{1/\gamma'}+\gamma^2\big)\log n+T\big)$ bits, and incurs an error probability of $O(\gamma/k+\beta_n)$.  Thus, we maintain a similar number of tests and decoding time as Theorem \ref{thm:gamma_main_theorem}, but the error probability increases, and in fact only behaves as $o(1)$ in the case that $\gamma=o(k)$ (which occurs, for example, under the mild condition $k = \Omega(\log n)$).

\textbf{Size-constrained tests:} Some care is required here to ensure that the constraints of our design matrix (i.e., fixed row and column weights) are satisfied. Specifically, at each level $l\in\{1,\dots,C\}$, we desire a hash function $h_l:\big\{1,\dots,\frac{n}{\rho^{1-l/C}}\big\}\rightarrow\{1,\dots,n/\rho\}$ such at each ``bucket'' (test) has a ``load'' (number of nodes in the test) of exactly $\rho^{l/C}$. An inspection of our analysis in Appendix \ref{sec:rho_algo_analysis} reveals that we only require the probability of two nodes colliding to be $O(\rho/n)$, i.e., only an approximately pairwise independent family is needed.

To construct the hash function above, we first consider a random permutation $\pi:\{1,\dots,\frac{n}{\rho^{1-l/C}}\}\rightarrow\{1,\dots,\frac{n}{\rho^{1-l/C}}\}$ such that for any $i,i'\in\{1,\dots,\frac{n}{\rho^{1-l/C}}\big\}$, we have $\mathbb{P}[|\pi(i)-\pi(i')|\leq t]=O\big(t\rho^{1-l/C}/n\big)$.  Such permutations are well-understood (e.g., see Definition 4.1 and Lemma 4.1 in \cite{Cev16}), and we can use this to design a hash function $h_l(\cdot)$ in the following manner: First apply the permutation discussed above, and then truncate the last $(l/C)\log_2\rho$ bits of the permutation value.  Then, for any $i,i'\in\{1,\dots,\frac{n}{\rho^{1-l/C}}\big\}$, we have
\begin{align}
    \mathbb{P}[h_l(i)=h_l(i')]
    \stackrel{(a)}{\le}\mathbb{P}[|\pi(i)-\pi(i')|\leq\rho^{l/C}]
    \stackrel{(b)}{=}O\Big(\frac{\rho}{n}\Big),
\end{align}
where (a) holds since if $i$ and $i'$ are in the same bucket, then all their bits except the last $(l/C)\log_2\rho$ bits are the same, and $\pi(i)$ and $\pi(i')$ can be at most $\rho^{l/C}$ (bucket size) apart; then, (b) holds by applying the collision property of our permutation. This proves that the constructed hash function has the required properties.  Moreover, we have $\mathsf{T}_{\text{hash}}=O(1)$ and $\mathsf{S}_{\text{hash}}=O(1)$. 

Given the preceding hash construction, we again provide a brief analysis of the storage as follows: Recall that we use $N=O(1)$ hashes at each level (except $l=0$ and $l=C$), and $C'=O(1)$ hashes at the final level, for a total of $O(1)$ hashes, requiring $O(\mathsf{S}_{\text{hash}}) = O(1)$ storage. In addition, under the high probability event that there are $O(k\rho^{1/C})$ possibly defective nodes at each level, their storage requires $O(k\rho^{1/C})$ integers, or $O(k\rho^{1/C}\log n)=o(n/\rho)$ bits (see Lemma \ref{lem:asymp}). Lastly, we need to store a total of $O(n/\rho)$ test outcomes, each requiring a bit of storage. Hence, the total storage is $O(\mathsf{S}_{\text{hash}}+k\rho^{1/C}\log n+n/\rho)=O(n/\rho)$ bits.

\textbf{Noisy setting:} Since we only need pairwise independence in our analysis in Appendix \ref{sec:noisy_algo_analysis}, we can use any pairwise independent hash family to generate a hash function, which only requires $\mathsf{T}_{\text{hash}}=O(1)$ and $\mathsf{S}_{\text{hash}}=O(\log n)$ (e.g., see \cite[Section 3.1]{Eri20}).   Here the analysis of the number of tests, error probability, and decoding time in Appendix \ref{sec:noisy_algo_analysis} remain unchanged. 

We provide a brief analysis of the storage as follows: Recalling our choices of $C',t, N=O(1)$, we use $N=O(1)$ hashes at each level except the last, and $C'N\log_2n=O(\log n)$ hashes at the final level, for a total of $O(\log n)$ hashes, requiring $O(\mathsf{S}_{\textup{hash}}\log n)$ storage. In addition, for any level $l$, we know that $|\mathcal{PD}^{(l)}|=O(k\log n)$ w.h.p, which implies that the storage required for the possibly defective set is $O(k\log n)$ integers, or $O(k\log^2n)$ bits. Lastly, we need to store a total of $O(k\log n)$ test outcomes, each requiring a bit of storage. The total storage is $O(\mathsf{S}_{\textup{hash}}\log n+k\log^2n+k\log n)=O(k\log^2n)$ by substituting $\mathsf{S}_{\text{hash}}=O(\log n)$.

\bibliographystyle{myIEEEtran}
\bibliography{refs}

\begin{thebibliography}{10}
\providecommand{\url}[1]{#1}
\csname url@samestyle\endcsname
\providecommand{\newblock}{\relax}
\providecommand{\bibinfo}[2]{#2}
\providecommand{\BIBentrySTDinterwordspacing}{\spaceskip=0pt\relax}
\providecommand{\BIBentryALTinterwordstretchfactor}{4}
\providecommand{\BIBentryALTinterwordspacing}{\spaceskip=\fontdimen2\font plus
\BIBentryALTinterwordstretchfactor\fontdimen3\font minus
  \fontdimen4\font\relax}
\providecommand{\BIBforeignlanguage}[2]{{%
\expandafter\ifx\csname l@#1\endcsname\relax
\typeout{** WARNING: IEEEtranS.bst: No hyphenation pattern has been}%
\typeout{** loaded for the language `#1'. Using the pattern for}%
\typeout{** the default language instead.}%
\else
\language=\csname l@#1\endcsname
\fi
#2}}
\providecommand{\BIBdecl}{\relax}
\BIBdecl

\bibitem{Ald14a}
M.~Aldridge, L.~Baldassini, and O.~Johnson, ``Group testing algorithms: Bounds
  and simulations,'' \emph{IEEE Trans. Inf. Theory}, vol.~60, no.~6, pp.
  3671--3687, June 2014.

\bibitem{Ald19}
M.~Aldridge, O.~Johnson, and J.~Scarlett, ``Group testing: An information
  theory perspective,'' \emph{Found. Trend. Comms. Inf. Theory}, vol.~15, no.
  3--4, pp. 196--392, 2019.

\bibitem{AVAL08}
J.-C. Aval, ``Multivariate {F}uss–{C}atalan numbers,'' \emph{Discrete
  Mathematics}, vol. 308, no.~20, pp. 4660 -- 4669, 2008.

\bibitem{Bay20}
W.~H. Bay, J.~Scarlett, and E.~Price, ``Optimal non-adaptive probabilistic
  group testing in general sparsity regimes,'' 02 2022.

\bibitem{Ber08a}
R.~Berinde, A.~C. Gilbert, P.~Indyk, H.~Karloff, and M.~J. Strauss, ``Combining
  geometry and combinatorics: A unified approach to sparse signal recovery,''
  in \emph{Allerton Conf. on Comm., Control and Comp.}, 2008.

\bibitem{Bon19a}
S.~Bondorf, B.~Chen, J.~Scarlett, H.~Yu, and Y.~Zhao, ``Sublinear-time
  non-adaptive group testing with {$O(k\log n)$} tests via bit-mixing coding,''
  \emph{IEEE Trans. Inf. Theory}, vol.~67, no.~3, pp. 1559--1570, 2020.

\bibitem{Cai13}
S.~Cai, M.~Jahangoshahi, M.~Bakshi, and S.~Jaggi, ``Efficient algorithms for
  noisy group testing,'' \emph{IEEE Trans. Inf. Theory}, vol.~63, no.~4, pp.
  2113--2136, 2017.

\bibitem{Cev16}
V.~Cevher, M.~Kapralov, J.~Scarlett, and A.~Zandieh, ``An adaptive
  sublinear-time block sparse {F}ourier transform,'' in \emph{ACM Symp. Theory
  Comp. (STOC)}, 2017.

\bibitem{Cha14}
C.~L. Chan, S.~Jaggi, V.~Saligrama, and S.~Agnihotri, ``Non-adaptive group
  testing: Explicit bounds and novel algorithms,'' \emph{IEEE Trans. Inf.
  Theory}, vol.~60, no.~5, pp. 3019--3035, May 2014.

\bibitem{Che09}
M.~Cheraghchi, ``Noise-resilient group testing: Limitations and
  constructions,'' in \emph{Int. Symp. Found. Comp. Theory}, 2009.

\bibitem{cher20}
M.~Cheraghchi and V.~Nakos, ``Combinatorial group testing and sparse recovery
  schemes with near-optimal decoding time,'' in \emph{IEEE Found. Symp. Comp.
  Sci. (FOCS)}, 2020.

\bibitem{Cor08}
G.~Cormode and M.~Hadjieleftheriou, ``Finding frequent items in data streams,''
  \emph{Proc. VLDB Endow.}, vol.~1, no.~2, p. 1530–1541, Aug. 2008.

\bibitem{Cor05a}
G.~Cormode and S.~Muthukrishnan, ``An improved data stream summary: The
  count-min sketch and its applications,'' \emph{J. Algs.}, vol.~55, no.~1, pp.
  58--75, 2005.

\bibitem{Cor06}
G.~Cormode and S.~Muthukrishnan, ``Combinatorial algorithms for compressed
  sensing,'' in \emph{Int. Colloq. Struct. Inf. Comm. Complex.}, 2006.

\bibitem{Dor43}
R.~Dorfman, ``The detection of defective members of large populations,''
  \emph{Ann. Math. Stats.}, vol.~14, no.~4, pp. 436--440, 1943.

\bibitem{Ven19}
V.~Gandikota, E.~Grigorescu, S.~Jaggi, and S.~Zhou, ``Nearly optimal sparse
  group testing,'' \emph{IEEE Trans. Inf. Theory}, vol.~65, no.~5, pp. 2760 --
  2773, 2019.

\bibitem{Oli20}
O.~Gebhard, M.~Hahn-Klimroth, O.~Parczyk, M.~Penschuck, M.~Rolvien,
  J.~Scarlett, and N.~Tan, ``Near optimal sparsity-constrained group testing:
  Improved bounds and algorithms,'' \emph{IEEE Trans. Inf. Theory}, vol.~68,
  no.~5, pp. 3253--3280, 2022.

\bibitem{Oli20a}
O.~Gebhard, O.~Johnson, P.~Loick, and M.~Rolvien, ``Improved bounds for noisy
  group testing with constant tests per item,'' \emph{IEEE Trans. Inf. Theory},
  vol.~68, no.~4, pp. 2604--2621, 2022.

\bibitem{Gil07}
A.~C. Gilbert, M.~J. Strauss, J.~A. Tropp, and R.~Vershynin, ``One sketch for
  all: Fast algorithms for compressed sensing,'' in \emph{ACM-SIAM Symp. Disc.
  Alg. (SODA)}, 2007.

\bibitem{Hogan2020}
C.~A. Hogan, M.~K. Sahoo, and B.~A. Pinsky, ``{Sample pooling as a strategy to
  detect community transmission of {SARS-CoV-2}},'' \emph{J. Amer. Med.
  Assoc.}, vol. 323, no.~19, pp. 1967--1969, 05 2020.

\bibitem{Hwa72}
F.~K. Hwang, ``A method for detecting all defective members in a population by
  group testing,'' \emph{J. Amer. Stats. Assoc.}, vol.~67, no. 339, pp.
  605--608, 1972.

\bibitem{Hus19}
H.~A. Inan, P.~Kairouz, and A.~Ozgur, ``Sparse combinatorial group testing,''
  \emph{IEEE Trans. Inf. Theory}, vol.~66, no.~5, pp. 2729--2742, 2020.

\bibitem{Ina19}
H.~A. {Inan}, P.~{Kairouz}, M.~{Wootters}, and A.~{\"Ozg\"ur}, ``On the
  optimality of the {K}autz-{S}ingleton construction in probabilistic group
  testing,'' \emph{IEEE Trans. Inf. Theory}, vol.~65, no.~9, pp. 5592--5603,
  Sept. 2019.

\bibitem{Ina20}
H.~A. Inan and A.~Ozgur, ``Strongly explicit and efficiently decodable
  probabilistic group testing,'' in \emph{IEEE Int. Symp. Inf. Theory (ISIT)},
  2020.

\bibitem{Ind10}
P.~Indyk, H.~Q. Ngo, and A.~Rudra, ``Efficiently decodable non-adaptive group
  testing,'' in \emph{ACM-SIAM Symp. Disc. Alg. (SODA)}, 2010.

\bibitem{Ind11}
P.~Indyk and E.~Price, ``K-median clustering, model-based compressive sensing,
  and sparse recovery for earth mover distance,'' in \emph{ACM Symp. Theory
  Comp. (STOC)}, 2011, pp. 627--636.

\bibitem{Joh16}
O.~Johnson, M.~Aldridge, and J.~Scarlett, ``Performance of group testing
  algorithms with near-constant tests-per-item,'' \emph{IEEE Trans. Inf.
  Theory}, vol.~65, no.~2, pp. 707--723, Feb. 2019.

\bibitem{Lar19}
K.~G. Larsen, J.~Nelson, H.~L. Nguyundefinedn, and M.~Thorup, ``Heavy hitters
  via cluster-preserving clustering,'' \emph{Comm. ACM}, vol.~62, no.~8, p.
  95–100, July 2019.

\bibitem{Lee16}
K.~Lee, R.~Pedarsani, and K.~Ramchandran, ``{SAFFRON}: A fast, efficient, and
  robust framework for group testing based on sparse-graph codes,'' in
  \emph{IEEE Int. Symp. Inf. Theory (ISIT)}, 2016.

\bibitem{Mal78}
M.~Malyutov, ``The separating property of random matrices,'' \emph{Math. Notes
  Acad. Sci. {USSR}}, vol.~23, no.~1, pp. 84--91, 1978.

\bibitem{Ngo11}
H.~Q. Ngo, E.~Porat, and A.~Rudra, ``Efficiently decodable error-correcting
  list disjunct matrices and applications,'' in \emph{Int. Colloq. Automata,
  Lang., and Prog. (ICALP)}, 2011.

\bibitem{Eri20}
E.~Price and J.~Scarlett, ``A fast binary splitting approach to non-adaptive
  group testing,'' in \emph{Int. Conf. Rand. Comp. (RANDOM)}, 2020.

\bibitem{Sca15b}
J.~Scarlett and V.~Cevher, ``Phase transitions in group testing,'' in
  \emph{Proc. ACM-SIAM Symp. Disc. Alg. (SODA)}, 2016.

\bibitem{Sca17b}
J.~Scarlett and V.~Cevher, ``Near-optimal noisy group testing via separate
  decoding of items,'' \emph{IEEE Trans. Sel. Topics Sig. Proc.}, vol.~2,
  no.~4, pp. 625--638, 2018.

\bibitem{Sca18b}
J.~Scarlett and O.~Johnson, ``Noisy non-adaptive group testing: A
  (near-)definite defectives approach,'' \emph{IEEE Trans. Inf. Theory},
  vol.~66, no.~6, pp. 3775--3797, 2020.

\bibitem{Nel20}
N.~Tan and J.~Scarlett, ``Near-optimal sparse adaptive group testing,'' in
  \emph{IEEE Int. Symp. Inf. Theory}, 2020.

\bibitem{Ver12}
R.~Vershynin, ``Introduction to the non-asymptotic analysis of random
  matrices,'' \emph{Compressed Sensing: Theory and Applications}, p. 210–268,
  2010.

\bibitem{Yelin2020}
I.~Yelin, N.~Aharony, E.~Shaer-Tamar, A.~Argoetti, E.~Messer, D.~Berenbaum,
  E.~Shafran, A.~Kuzli, N.~Gandali, T.~Hashimshony, Y.~Mandel-Gutfreund,
  M.~Halberthal, Y.~Geffen, M.~Szwarcwort-Cohen, and R.~Kishony, ``Evaluation
  of {COVID}-19 {RT}-q{PCR} test in multi-sample pools,'' vol.~71, no.~16, pp.
  2073--2078, 2020.

\end{thebibliography}
 
\end{document}